\documentclass{article}

\usepackage{arxiv}

\usepackage[utf8]{inputenc} % allow utf-8 input
\usepackage[T1]{fontenc}    % use 8-bit T1 fonts
\usepackage{hyperref}       % hyperlinks
\usepackage{url}            % simple URL typesetting
\usepackage{booktabs}       % professional-quality tables
\usepackage{amsfonts}       % blackboard math symbols
\usepackage{nicefrac}       % compact symbols for 1/2, etc.
\usepackage{microtype}      % microtypography
\usepackage{lipsum}		% Can be removed after putting your text content
\usepackage{graphicx}

\usepackage{doi}

\usepackage[section]{placeins}

\usepackage{textalpha}
\usepackage{amssymb}
\usepackage{subcaption}
\usepackage{longtable}
\usepackage[ruled,vlined]{algorithm2e}
\usepackage{pgfplots}
\pgfplotsset{compat=newest}
\usepgfplotslibrary{fillbetween}
\usepackage{textalpha}
\usepackage{xcolor}

\usepackage[T1]{fontenc} % Use 8-bit encoding that has 256 glyphs

\usepackage[utf8]{inputenc} % Required for including letters with accents

\usepackage{graphicx} % Required for including images
\graphicspath{{Figures/}} % Set the default folder for images

\usepackage{enumitem} % Required for manipulating the whitespace between and within lists

\usepackage{lipsum} % Used for inserting dummy 'Lorem ipsum' text into the template
\usepackage{amsmath,amsthm} % For including math equations, theorems, symbols, etc

\usepackage{varioref} % More descriptive referencing

\usepackage{tikz}
\usetikzlibrary{3d}
\usetikzlibrary{arrows,shapes,positioning,calc}
\theoremstyle{definition} % Define theorem styles here based on the definition style (used for definitions and examples)

\theoremstyle{plain} % Define theorem styles here based on the plain style (used for theorems, lemmas, propositions)
\newtheorem{theorem}{Theorem}
\newtheorem{lemma}[theorem]{Lemma}

\theoremstyle{remark} % Define theorem styles here based on the remark style (used for remarks and notes)
\newtheorem*{remark}{Remark}

\newcommand{\R}{\mathbb{R}}

\newcommand{\x}{x}

\newcommand{\eps}{\epsilon}
\newcommand{\Omeps}{\Omega_{\eps}}
\newcommand{\Om}{\Omega}

\newcommand{\nv}{\vec{n}}
\newcommand{\Vv}{\vec{V}}
\newcommand{\uv}{u}

\newcommand{\Qv}{\vec{Q}}

\newcommand{\nab}{\nabla}

\newcommand{\Ga}{\Gamma_1}
\newcommand{\Gb}{\Gamma_2}
\newcommand{\Gc}{\Gamma_3}

\newcommand*\diff{\mathop{}\!\mathrm{d}}

\newcommand{\expnumber}[2]{{#1}\mathrm{e}{#2}}
\newcommand\restr[2]{{% we make the whole thing an ordinary symbol
		\left.\kern-\nulldelimiterspace % automatically resize the bar with \right
		#1 % the function
		\vphantom{\big|} % pretend it's a little taller at normal size
		\right|_{#2} % this is the delimiter
}}

\newcommand{\N}{\mathbb{N}}
\title{Shape Optimization for the Mitigation of Coastal Erosion via Smoothed Particle Hydrodynamics}

\author{ \href{https://orcid.org/0000-0001-9930-065X}{\includegraphics[scale=0.06]{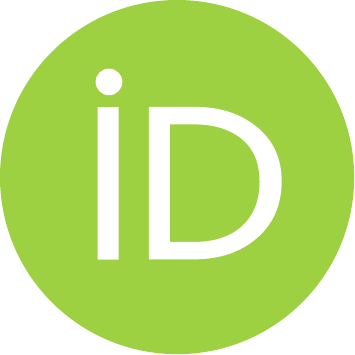}\hspace{1mm}Luka Schlegel} \\
	Department of Mathematics\\
	Universität Trier\\
	Universitätsring 15, 54296 Trier\\
	\texttt{schlegel@uni-trier.de} \\
	\And
	\href{https://orcid.org/0000-0001-7665-130X}{\includegraphics[scale=0.06]{orcid.pdf}\hspace{1mm}Volker Schulz} \\
	Department of Mathematics\\
	Universität Trier\\
	Universitätsring 15, 54296 Trier\\
	\texttt{volker.schulz@uni-trier.de} \\
}

\renewcommand{\headeright}{}
\renewcommand{\undertitle}{}
\renewcommand{\shorttitle}{SOMICE}

\hypersetup{
pdftitle={SPH Paper},
pdfsubject={q-bio.NC, q-bio.QM},
pdfauthor={Luka Schlegel, Volker Schulz},
pdfkeywords={Shape Optimization, SPH, Coastal Erosion},
}

\begin{document}
\maketitle

\begin{abstract}
	Adjoint-based shape optimization most often relies on Eulerian flow field formulations. However, since Lagrangian particle methods are the natural choice for solving sedimentation problems in oceanography, extensions to the Lagrangian framework are desirable. For the mitigation of coastal erosion, we perform shape optimization for fluid flows, that are described by Lagrangian shallow water equations and discretized via smoothed particle hydrodynamics. The obstacle's shape is hereby optimized over an appropriate cost function to minimize the height of water waves along the shoreline based on shape calculus. Theoretical results will be numerically verified exploring different scenarios.
\end{abstract}

% keywords can be removed

\section{Introduction}
Coastal erosion describes the displacement of land caused by destructive sea waves, currents or tides. Major efforts have been made to mitigate these effects using groins, breakwaters and various other structures. Numerical investigations of useful shapes require accurate descriptions of propagating waves and sediment. Lagrangian particle movements appear to be the natural choice for solving sedimentation problems in oceanography, being generally free of diffusion. In contrast, inherent numerical diffusion in Eulerian methods can only affect the resolution of purely advective flows. It is therefore tempting to investigate shape optimization techniques for Lagrangian fluid flows. Typically, associated meshfree particle methods can be divided in two groups, the ones that approximate the weak form, e.g. the diffusive element \cite{Nayroles1992}, the element-free Galerkin \cite{Belytschko1994} or the reproducing kernel method \cite{Liu1995}, such as the ones that approximate the strong form, e.g. the moving particle \cite{Koshizuka1996} or the vortex method \cite{Ogami1991}. In this chapter we investigate a particle flow that falls into the latter class - the smoothed particle hydrodynamics (SPH) \cite{Monaghan2002}. Both groups share the property that shape optimization has been only rarely investigated. In general most approaches are based on the weak form, introduced in order to enable large deformations, e.g. for the element-free Galerkin method \cite{Bobaru2002} or the reproducing kernel particle method \cite{Grindeanu1998}\cite{Chen2007}. Lately a one-way coupled, volume averaged transport model was used to circumvent the direct usage of the Lagrangian particle flow \cite{Hohmann2019}. In particular for SPH fluids only boundary contributions via ghost particles and the direct differentiation method to obtain optimal fluid-structure interactions have been investigated \cite{Ha2011}. However, this method comes with the drawback of solving an additional problem for each design parameter, which can become costly for highly resolved obstacles. Hence, adjoint-based shape optimization for SPH fluids and accordingly also the usage in mitigation of coastal erosion appears novel. \\
The paper is structured as follows: We derive adjoints for a general class of particle systems in Section \ref{Sec:AdjointParticleSystems}, before Section \ref{Sec:SPH} will build the foundation for discretizing shallow water equations (SWE) via SPH-fluids, restructuring rigid and outflow boundaries via signed-distance fields \cite{Koschier2017} for the use in optimization. Subsequently, Section \ref{Sec:SOSPH} will derive shape derivatives based on the adjoints that have been developed for the general case. The technique boils down to derive shape derivatives for a finite element interpolator building up on discrete investigations \cite{Berggren2010}\cite{Schneider2008}. The results are numerically verified in simplified scenarios using a gradient-descent algorithm. In this setting, partial derivatives required for recursive adjoints are calculated via automatic differentiation and a deformation gradient is obtained from linear elasticity evaluations \cite{Schulz2016}. 

\section{Adjoint for Particle Systems}
\label{Sec:AdjointParticleSystems}
In this section adjoints for a system of $d-$dimensional particles consisting of position and velocities $X_k=(\x_k,\uv_k)\in\R^{dN}\times\R^{dN}$ for time steps $k\in\{1,...,n\}$ with constant particle mass $m>0$ are derived. The iteration laws are described by a symplectic Euler, i.e. for force functions $F:\R^{dN}\times\R^{dN}\times \Psi\rightarrow\R^{dN}$ and to be determined control $q\in\Psi$. We solve
	\begin{equation}
	\begin{aligned}
	\uv_{k}&=\uv_{k-1}+\frac{\Delta t}{m}F(\uv_{k-1},\x_{k-1},q)\\
	\x_{k}&=\x_{k-1}+\Delta t\uv_{k}.
	\label{Eq12OriginalSymplecticEuler}
	\end{aligned}			
	\end{equation}
The solution of the particle system constrains a time-dependent objective function $J:\R^{ndN}\times\R^{ndN}\times \Psi\rightarrow\R$, i.e.
\begin{equation}
\begin{aligned}
J(\uv,\x,q)=\sum_{k=1}^{n}J_k(\uv_k,\x_k,q)\text{.}
\label{Eq:ObjAdj}
\end{aligned}			
\end{equation}
In this setting, we can derive the following theorem, consisting of recursively defined adjoints and an equation for the sensitivity calculation with respect to control.
\begin{theorem}
(Adjoint Particle System)
	\label{Theo:AdjointsParticle}
	Assume a particle system is solved iteratively using iteration (\ref{Eq12OriginalSymplecticEuler}), then sequences $\{\delta_k\}_{k=1}^n$ and $\{\mu_k\}_{k=1}^n$ required in adjoint computations are obtained from backward recursion, i.e.
	\begin{align}
	\delta_k&=\delta_{k+1}+\frac{\Delta t}{m}(F^{\x_k})^T\mu_{k+1}+\frac{\Delta t^2}{m}(F^{\x_k})^T\delta_{k+1}+\left(J_k^{\x_k}\right)^T\\
	\mu_k&=\mu_{k+1}+\frac{\Delta t}{m}(F^{\uv_k})^T\mu_{k+1}+\frac{\Delta t}{m}\delta_{k+1}+\frac{\Delta t^2}{m}(F^{\uv_k})^T\delta_{k+1}+\left(J_k^{\uv_k}\right)^T\text{,}
	\label{Eq:AdjointRecursion}
	\end{align}
	where the recursion starts at
	\begin{align}
\mu_n=J_n^{\uv_{n}} \quad \delta_n=J_n^{\x_{n}}\text{.}
\end{align}
	The sensitivity of the objective function with regards to the control $q$ are calculated as 
	\begin{align}
J^q=\sum_{k=1}^{n}J_k^q+\frac{\Delta t}{m}(F_{k-1}^q)^T\mu_k+\frac{\Delta t^2}{m}(F_{k-1}^q)^T\delta_k\text{.}
\label{Eq:Sensitivity}
\end{align}
\end{theorem}

\begin{proof}
	For simplicity we will rewrite the symplectic Euler by substitution of $u_k$ in (\ref{Eq12OriginalSymplecticEuler})
	\begin{equation}
	\begin{aligned}
	\uv_{k}&=\uv_{k-1}+\frac{\Delta t}{m}F(\uv_{k-1},\x_{k-1},q)\\
	\x_{k}&=\x_{k-1}+\Delta t\left(\uv_{k-1}+\frac{\Delta t}{m}F(\uv_{k-1},\x_{k-1},q)\right)\text{.}
	\label{Eq13RewrittenSymplecticEuler}
	\end{aligned}			
	\end{equation}
	The idea is related to \cite{Schulz1996} for adjoint calculations of iteration equations.
We first sum up the products of iterates and unknown multipliers together with the objective (\ref{Eq:ObjAdj})
		\begin{equation}
		\begin{aligned}
		0=&\sum_{k=1}^{n}\left[	\uv_{k}-\uv_{k-1}-\frac{\Delta t}{m}F(\uv_{k-1},\x_{k-1},q)\right]^T\mu_k\\
		+&\sum_{k=1}^{n}\left[\x_{k}-\x_{k-1}+\Delta t\left(\uv_{k-1}-\frac{\Delta t}{m}F(\uv_{k-1},\x_{k-1},q)\right)\right]^T\delta_k\\
		+&J-\sum_{k=1}^{n}J_k\left(\uv_k,\x_k,q\right)\text{.}
		\end{aligned}
\end{equation}
	Differentiating w.r.t. the control leads to
	\begin{equation}
\begin{aligned}
0=&\sum_{k=1}^{n}\left[	\uv_{k}^q-\uv_{k-1}^q-\frac{\Delta t}{m}F^{\uv_{k-1}}\uv_{k-1}^q-\frac{\Delta t}{m}F^{\x_{k-1}}\x_{k-1}^q-\frac{\Delta t}{m}F_{k-1}^q)\right]^T\mu_k\\
+&\sum_{k=1}^{n}\left[	\x_{k}^q-\x_{k-1}^q+\Delta t\left(\uv_{k-1}^q-\frac{\Delta t}{m}F^{\uv_{k-1}}\uv_{k-1}^q-\frac{\Delta t}{m}F^{\x_{k-1}}\x_{k-1}^q-\frac{\Delta t}{m}F_{k-1}^q)\right)\right]^T\delta_k\\
+&J^q-\sum_{k=1}^{n}\left(J_k^{\uv_k}\right)^T\uv_k^q+\left(J_k^{\x_k}\right)^T\x_k^q+J_k^q\text{.}
\end{aligned}
\end{equation}
	Using $\x_0^q=\uv_0^q=0$, shifting indices in sums and reorder terms for the unknowns $\x_k^q,\uv_k^q$ for $k\in\{1,...,n\}$
		\begin{equation}
\begin{aligned}
0=&\sum_{k=1}^{n-1}\left[\mu_k-\mu_{k+1}-\frac{\Delta t}{m}(F^{\uv_k})^T\mu_{k+1}-\frac{\Delta t}{m}\delta_{k+1}-\frac{\Delta t^2}{m}(F^{\uv_k})^T\delta_{k+1}-J_k^{\uv_k}\right]^T\uv_k^q\\
+&\sum_{k=1}^{n-1}\left[-\frac{\Delta t}{m}(F^{\x_k})^T\mu_{k+1}+\delta_k-\delta_{k+1}-\frac{\Delta t^2}{m}(F^{\x_k})\delta_{k+1}-J_k^{\x_k}\right]^T\x_k^q\\
+&\left[\mu_n-J_n^{\uv_n}\right]^T\uv_n^q+\left[\delta_n-J_n^{\x_n}\right]^T\x_n^q\\
+&J^q-\sum_{k=1}^{n}J_k^q-\frac{\Delta t}{m}(F_{k-1}^q)^T\mu_k-\frac{\Delta t^2}{m}(F_{k-1}^q)^T\delta_k\text{.}
\end{aligned}
\end{equation}
	The recursion is finally obtained from the coefficients before the unknown partial derivatives, the final conditions follow from the third line. The sensitivity of the objective function with regards to the control $q$ follows from the last line.
\end{proof}

\begin{remark}
  	In this paper we will be concerned with a total of $N$ particles in two dimensions such that $\x_k,\uv_k\in\R^{2N}$. The recursion (\ref{Eq:AdjointRecursion}) can be written as
	\begin{equation}
\begin{aligned}
\mu_k&=\left(\textbf{I}_{2N}+\frac{\Delta t}{m}F^{\uv_k}\right)^T\left(\mu_{k+1}+\Delta t\delta_{k+1}\right)+\left(J_k^{\uv_k}\right)^T\\
\delta_k&=\delta_{k+1}+\left(\frac{\Delta t}{m}F^{\x_k}\right)^T\left(\mu_{k+1}+\Delta t\delta_{k+1}\right)+(J_k^{\x_k})^T\text{.}
\end{aligned}
\label{Eq24NParticlesDDimAdjoint}			
\end{equation}
	In this form obtained adjoints equal the ones that have been derived earlier by different means \cite{McNamara2004,Wojtan2006,Schneider2008}.
\end{remark}
\begin{remark}
	In order to control a particle system, we only require the ability to calculate the matrix with partial derivatives of the forces. Further sensitivity calculations need to specify the particle system such as control $q$. We will restrict to the aforementioned SPH fluids in the following section, where the shape of an obstacle is optimized. The sensitivity is obtained in Section \ref{SubSec:DerSDSPH} by evaluating the discrete shape derivative $DJ(\Om)[\Vv_l]$ for domain $\Om$ in direction $V_l$ via formula (\ref{Eq:Sensitivity}).
\end{remark}

\section{Smoothed Particle Hydrodynamics}
\label{Sec:SPH}
Before being concerned with optimization of the SPH flows, Section \ref{SubSec:SPHBasics} will discuss the basic idea, definitions and notations for this technique, while Section \ref{SubSec:SPHNumBoundary} is dealing with boundary interactions of SPH particles, that are crucial in shape sensitivity calculations.
\subsection{Basics of SPH}
\label{SubSec:SPHBasics}
Central for SPH-particles travelling on a domain $\Om\subset\R^d$ is the approximation of the delta distribution by the usage of kernels. Hence, a field value $v:\Om\rightarrow\R$ is written using the symbolical Dirac-delta identity and the Dirac-delta condition \cite{Monaghan2002}
\begin{align}
v(x)=\int v(x')\delta(x-x')\diff x'=\lim_{h\rightarrow0}\int v(x')W(x-x',h)\diff x'\text{,}
\end{align}
where $W:\R^d\times\R_+\rightarrow\R$ is defined as a valid kernel if fulfilling properties of normalization, Dirac-delta limiting such as often positivity, symmetry and compactness of the support \cite{Monaghan2002}.
In turn this is intended to serve as limiting object in the sense of
\begin{align}\label{Eq:SPHIdea}
\lim_{h\rightarrow0}\int v(x')W(x-x',h)\diff x'=\lim_{h\rightarrow0}\lim_{N\rightarrow\infty}\sum_{i=1}^{N}m^i\frac{v^i}{\rho^i}W(x-x^i,h)
\end{align}
for constant particle density $\rho^i$ and particle mass $m^i$.
Ultimately, the latter is approximated for fixed $h>0$ and $N\in\N$, forming the basis of SPH-techniques.
In the literature various valid kernel functions are defined, we restrict ourselves here to the original one, hence define the Gaussian kernel as \cite{Monaghan1985}
	\begin{align}
	W(x-x^i,h)=\sigma_g\exp\left[-\frac{||x-x^i||^2}{h^2}\right]
	\label{Eq:GaussianKernel}
	\end{align}
	e.g. with two-dimensional normalization as
	\begin{equation}
			\begin{aligned}
		\sigma_g&=\frac{1}{\pi h^2}
		.
		\end{aligned}
	\end{equation}

\begin{remark}
	The Gaussian kernel comes with the advantage of belonging to class $C^\infty$, we will use this fact in Section \ref{Sec:SOSPH}.
	However, it lacks compact support, which results in a summation over all particles in the domain. To counter this, often cut-off kernels as the Cubic Spline \cite{Monaghan2002} or Poly6 \cite{Muller2003} kernel are used. In classical form these kernels are not differentiable at the cut-off, but are attractive from computational side, as the summation is only performed over particles in the limited support radius.
\end{remark}
\begin{remark}
	A limited number of publications have been dealing with the convergence of SPH methods, based on joint particle limits $N\rightarrow\infty$ and smoothing limit $h\rightarrow0$ as in (\ref{Eq:SPHIdea}). First attempts relied on spatial discretizations with time-continuous approximations \cite{Raviart1985} or at least the knowledge about exact particle trajectories \cite{Moussa2000}. In addition, convergence results have been presented for consecutive limits of discretization parameter and smoothing radius \cite{Dilisio1998} or a selective choice of kernel functions \cite{Oelschlaeger1991}. 
\end{remark}

\subsection{Particle Boundary Interactions}
\label{SubSec:SPHNumBoundary}
Since we are ultimately interested in shape optimization of particle systems, we need to specify the particle boundary interaction. In classical SPH methods, this is typically done via boundary particles \cite{Monaghan2002}. The obvious drawback are increased computational efforts, undesired boundary frictions such as the inability to model complex geometries \cite{Koschier2017}.\\
All SPH boundary techniques have the common idea to approximate the second integral below, which arises naturally when being restricted to a bounded domain $\Om\subset D\subset\R^d$
\begin{align}
\rho(x)&=\lim_{h\rightarrow 0}\int_{\Om}W(x-x',h)\rho(x')\diff x'+\lim_{h\rightarrow 0}\int_{D\setminus\Om} W(x-x',h)\rho(x')\diff x'\text{.}
\label{Eq:FluidBoundarySplit}
\end{align}
In the last decade various researchers addressed this problem e.g. by the usage of continuous boundary methods such as the boundary surface integral \cite{Kulsegaram2004} or the signed distance field method \cite{Koschier2017}. This work will restrict to latter ideas and extends them for Lagrangian SWE with rigid and outflow boundary conditions. In \cite{Koschier2017} the boundary density portion $\rho_D$ is approximated as
\begin{align}
\rho_D(x)&=\lim_{h\rightarrow 0}\int_{D\setminus\Om}\gamma(d_\Om(x')) W(x-x',h)\diff x'\text{,}
\label{Eq:SPHBoundaryContr}
\end{align}
where a modification of the signed distance function 
\begin{align}
d_{\Om}(x)=\begin{cases}
	d(x,\partial\Om)\quad &\text{ if } x\in \Om\\
	0 \quad &\text{ if } x\in \partial\Om
	\\
	-d(x,\partial\Om)\quad &\text{ if } x\in D\setminus\Om
\end{cases}
\label{Eq:SPHSDF}
\end{align}
is used with
\begin{align}
\gamma(d_\Om)=\begin{cases}
\rho_0(1-\frac{d_\Om}{h}) &\text{ if }d_\Om\leq h\\
0 &\text{ otherwise}
\end{cases}
\label{Eq:ModSDF}
\end{align}
for reference density $\rho_0$. 
In addition to rigid boundaries, we implement open-sea boundaries. Due to this reason, we extend the idea of buffered layers for the modelling of outflow conditions, as introduced in \cite{Vacondio2012} for SPH-based computations, to mesh-based signed distance maps. The idea is to create a collecting channel, in which the particle movement is decelerated. Suppose an additional layer as subdomain $\Om_L\subset\R^d$ is introduced, that builds with $\Om$ a conformal domain. Then as in (\ref{Eq:SPHSDF}), we require
\begin{align}
d_{\Om_L}(x)=\begin{cases}
d(x,\partial\Om_L)\quad &\text{ if } x\in (\Om\cup D)\\
0 \quad &\text{ if } x\in \partial\Om
\\
-d(x,\partial\Om_L)\quad &\text{ if } x\in \Om_L\setminus(\Om\cup D).
\end{cases}
\label{Eq:SPHSDF1}
\end{align}
Hence, outflow boundary conditions are created via subdomain-dependent modification of this very signed distance function, i.e.
\begin{align}
\gamma^L(d_{\Om_L})=
\begin{cases}
\rho_0d_{\Om_L} &\text{ if } x\in\Om_L\\
0 &\text{ otherwise}.
\end{cases}
\label{Eq:ModSDF1}
\end{align}
In addition, we add a decelerating domain-dependent coefficient for the velocity in the particles iteration law of (\ref{Eq12OriginalSymplecticEuler}), i.e. for $\eps>0$
\begin{align}
\phi:=\begin{cases}
\phi_1=1 \text{ in } \Om\cup D\\
\phi_2=\eps \text{ in } \Om_L\text{.}
\end{cases}
\end{align}
The implementation in the discretized setting is done  via interpolation for each particle and its position $\{\x^i_k\}_{i,k=1}^{N,n}$ on a finite element mesh, e.g. consisting of Lagrangian elements $\kappa$, with $S=\dim(\mathbb{P}_p)=\frac{(d+p)!}{d!p!}$ nodal degrees of freedom for $d$ dimensions, for polynomial space $P:=\mathbb{P}_p$ of polynomial order $p\in\mathbb{N}$ with associated shape functions $\{N_1^p,...,N_{S}^p\}$, i.e. 
\begin{align}
\rho_{D}(x)=\sum_{s=1}^{S}m\gamma_{s}(d_\Om)N_s^p(x)+\sum_{s=1}^{S}m\gamma^L_{s}(d_{\Om_L})N_s^p(x)\text{.}
\label{Eq:BoundaryFI}
\end{align}
The interpolation technique is then also used to compute the boundary forces by
\begin{align}
F_{D}^{Height}(x)=\nab\rho_{D}(x)=\sum_{s=1}^{S}\gamma_{s}(d_\Om)\nab N_s^p(x)+\sum_{s=1}^{S}\gamma^L_{s}(d_{\Om_L})\nab N_s^p(x).
\label{Eq:ParticleBoundaryPressure}
\end{align}

\begin{remark}
	The advantage in using this boundary representation lies in the possibility to precompute a solution field, allowing cheap finite element interpolations, whenever a particle is in the proximity of the boundary.\\
	 Later on the sediment field $z:\Om\rightarrow\R$ of (\ref{Eq13:Lagconsmom}) can be created in the same manner as the boundary density. In this setting, we can naturally identify domain boundaries via increased sediment elevations.
\end{remark}
\begin{remark}
	Rigid boundary contributions, as in (\ref{Eq:SPHBoundaryContr}), rely on the max-function, i.e. 
	\begin{align*}
	\gamma(d_\Om(x))=\rho_0\left(\max\left\{0,1-\frac{d_\Om(x)}{h}\right\}\right),
	\end{align*}
	since this is not differentiable for $x\in\Om$ such that $d_\Om(x)=h$, we will use a smoothed $\max$-function $C^1(\Om)\ni\max_\alpha:\Om\rightarrow\R$ in the following section e.g. as \cite{Christof2017}
	\begin{align}
	\text{max}_\alpha(d_\Om)=\begin{cases}
	\max(0,1-\frac{d_\Om}{h}) &\text{for } 1-\frac{d_\Om}{h}\in\R\setminus\left[-\frac{1}{\alpha},\frac{1}{\alpha}\right]\\
	\frac{\alpha}{4}(1-\frac{d_\Om}{h})^2+\frac{1}{2}(1-\frac{d_\Om}{h})+\frac{1}{4\alpha}&\text{otherwise}
	\end{cases}
	\label{Eq:Smoothed}
	\end{align}
	for $\alpha>0$ with derivative as
	\begin{align}
	\text{max}'_\alpha(d_\Om)=\begin{cases}
	0 &\text{for } 1-\frac{d_\Om}{h}\in\left(-\infty,-\frac{1}{\alpha}\right)\\
	-\frac{\alpha}{2h}(1-\frac{d_\Om}{h})-\frac{1}{2h}&\text{for } 1-\frac{d_\Om}{h}\in\left[-\frac{1}{\alpha},\frac{1}{\alpha}\right]
	\\
	-\frac{1}{h}&\text{for } 1-\frac{d_\Om}{h}\in\left(\frac{1}{\alpha},\infty\right)\text{.}
	\end{cases}
	\end{align}
	For outflow boundaries we rely on an analogous $C^1$-counterpart of the $\min$-function as 
		\begin{align}
	\text{min}_\alpha(d_{\Om_L})=\begin{cases}
	\min(0,d_{\Om_L}) &\text{for } d_{\Om_L}\in\R\setminus\left[-\frac{1}{\alpha},\frac{1}{\alpha}\right]\\
	-\frac{\alpha}{4}d_{\Om_L}^2+\frac{1}{2}d_{\Om_L}-\frac{1}{4\alpha}&\text{otherwise}
	\end{cases}
	\label{Eq:SmoothedMin}
	\end{align}
\end{remark}
with
	\begin{align}
\text{min}'_\alpha(d_\Om)=\begin{cases}
1 &\text{for } d_{\Om_L}\in\left(-\infty,-\frac{1}{\alpha}\right)\\
-\frac{\alpha}{2}(d_{\Om_L})+\frac{1}{2}&\text{for } d_{\Om_L}\in\left[-\frac{1}{\alpha},\frac{1}{\alpha}\right]
\\
0&\text{for } d_{\Om_L}\in\left(\frac{1}{\alpha},\infty\right)\text{.}
\end{cases}
\end{align}

\section{Adjoint-Based Shape Optimization for SPH Particles}
\label{Sec:SOSPH}
The following section is devoted to the derivation of the shape derivative of an SPH-fluid with suitable boundary interaction for a model described in Section \ref{Sec:MoFo}. For this, necessary definitions and notations are recalled in Section \ref{SubSec:BasicsShapeOpt} and applied in Section \ref{SubSec:DerSDSPH}. A numerical verification of results for two test cases follows in Section \ref{Sec:NumericalExamplesSDSPH}. 

\subsection{Model Formulation}
\label{Sec:MoFo}
Suppose we are given an open domain $\Om\subset\R^2$, which is split into disjoint subdomains, consisting of a connected, interior domain $\Om_1$ such that $\Ga\cup\Gb\cup\Gc:=\Gamma:=\partial\Om_1$, a simply connected obstacle domain $D$ and an exterior domain $\Om_2\cup\Om_3:=\Om\setminus\bar{\Om}_1\setminus D$, such that $\bar{\Om}_1\cup \Om_2\cup\Om_3\cup D=\Om$. 
We assume the variable, interior boundary $\Gamma_3$ and the fixed outer $\Ga\cup\Gb$ of interior domain $\Om_1$ to be at least Lipschitz. One simple example of such kind is visualized below in Figure \ref{fig:particledomain}.

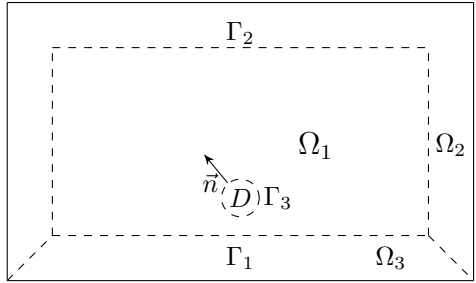
\begin{figure}[h]
	\centering
	\begin{tikzpicture}
%inner circle
\draw[->, >=stealth', shorten >=1pt] (2.33,0.7)   -- (2.0,1.1);
\draw[dashed] (0,0) --(5,0);
\draw[dashed] (0,0) --(0,2.5);
\draw[dashed] (5,0) --(5,2.5);
\draw[dashed] (0,2.5) --(5,2.5);
\draw[dashed] (2.5,0.5) circle (0.25);
\def\rk{-0.12}

%outer circle
\def\ra{-0.6}
\def\rb{5.6}
\def\rc{3.1}
\draw[dashed] (\ra,\ra) -- (0,0);
\draw[dashed] (5,0) -- (\rb,\ra);
\draw (\ra,\ra) --(\rb,\ra);
\draw (\ra,\ra) --(\ra,\rc);
\draw (\rb,\ra) --(\rb,\rc);
\draw (\ra,\rc) --(\rb,\rc);

\node (A) at (3.5,1.2) {\large $\Om_1$};
\node (B) at (2.5,-0.3) {$\Ga$};
\node (E) at (2.5,2.7) {$\Gb$};
\node (F) at (3,0.5) {$\Gc$};
\node (G) at (2.5,0.5) {$D$};
\node (H) at (2.1,0.7) {$\nv$};
\node (I) at (5.3,1.2) {$\Om_2$};
\node (J) at (4.5,-0.3) {$\Om_3$};
\end{tikzpicture}
	\caption[Illustrative Domain]{Illustrative Domain $\Om$ with Initial Circled Obstacle $D$, Interior Domain $\Om_1$, Exterior Domains $\Om_2,\Om_3$ and Boundaries $\Ga,\Gb,\Gc$}
	\label{fig:particledomain}
\end{figure}
On this domain fluids will follow the laws of the Lagrangian viscous SWE, i.e. on $\Om\times(0,T)$ we have
	\begin{align}
DH&=-H\nab\cdot \Qv\label{Eq12:Lagconsmass}\\
D\Qv&=-g\nab(H+z)+\mu\nab^2\Qv.\label{Eq13:Lagconsmom}
	\end{align}
subject to rigid boundary conditions on $\Gamma_1,\Gamma_3\times(0,T)$, open boundary conditions on $\Gamma_2\times(0,T)$ and suitable initial conditions on $\Om\times\{0\}$ for solution $U:\Om\times(0,T)\rightarrow\R\times\R^2$, where for simplicity the domain and time-dependent components are denoted by $U=(H,\Qv)=(H,Hu,Hv)$, with $H$ being the water height and $Hu,Hv$ the weighted horizontal and vertical discharge or velocity. The sediment is a scalar field $z:\Om\rightarrow\R$, $g$ the gravitational acceleration and $\mu>0$ the viscosity weight .\\
Our objective $J:\Om\rightarrow\R$ consists of three parts
\begin{align}
    J(\Om)=J_1(\Om)+J_2(\Om)+J_3(\Om)
    .
\end{align}
We interpret $\Ga$ as coastline and implement a tracking-type objective \cite{Reyes2011} for the rest height of the water, i.e.
\begin{align}
J_1(\Om)=\int_0^T\int_{\Ga}\frac{1}{2}(H-\bar{H})^2\diff s \diff t\text{.}
\label{Eq:Objective_Height}
\end{align}
The objective is accompanied by a volume penalty as a second part
\begin{align}
J_2(\Om)=\nu_1\int_{D}1\diff x
\label{Eq:1Vol}
\end{align}
and a perimeter regularization
\begin{align}
J_3(\Om)=\nu_2\int_{\Gamma_{3}}1\diff s
.
\label{Eq:1Peri}
\end{align}
The contribution of the penalty terms (\ref{Eq:1Vol}) and (\ref{Eq:1Peri}) is controlled by parameters $\nu_1,\nu_2>0$, which need to be defined a priori. Finally, with constraints of type (\ref{Eq12:Lagconsmass})-(\ref{Eq13:Lagconsmom}) we obtain a PDE-constrained optimization problem, which is intended to be solved in a discretize-then-differentiate setting. The discretization is based on SPH as described in Section \ref{SubSec:SPHBasics} with boundary interaction as in Section \ref{SubSec:SPHNumBoundary}. From this we obtain a particle system, which provides us in two dimensions with states via tuple $(x_k,v_k)\in\R^{2N}\times\R^{2N}$. Here we interpret $\Om_1$ as fluid and $\Om_2\cup\Om_3\cup D$ as boundary domain for the particles in reference to Section \ref{SubSec:SPHNumBoundary}.
Based on this we define the discrete counterpart to objective (\ref{Eq:Objective_Height}), i.e.
\begin{equation}
\begin{aligned}
J_{1,h}(\Om)&=\sum_{k=1}^{n}J_{1,h,k}(x_k,v_k,\Om)\\
&=\sum_{(k,x)\in\{1,...,n\}\times\Gamma_{1,h}}\frac{1}{2}\Delta t\left(\frac{\rho(x)}{\rho_0}-\bar{H}(x)\right)^2\\&=\sum_{(k,x,j)\in\{1,...,n\}\times\Gamma_{1,h}\times\{1,...,N\}}\frac{1}{2}\Delta t\left(\frac{m}{\rho_0}W(x-x^j_k,h)-\bar{H}(x)\right)^2\text{.}
\label{Eq:ParticleObjective}
\end{aligned}
\end{equation}
where $x_k^j$ defines the position of particle $j$ at time $k$, that is obtained as solution to an SPH-discretization of the Lagrangian SWE. 
 As it can be seen in (\ref{Eq:SPHIdea}) SPH methods allow for simplified density calculations, to utilize this benefit a density-water relation is used \cite{Solenthaler2011}, i.e.
\begin{align}
H=\frac{\rho}{\rho_0}\label{Eq:DensityWater}
\end{align}
for $\rho_0>0$.
This definition gives only local conservation results, such that (\ref{Eq12:Lagconsmass}) and (\ref{Eq13:Lagconsmom}) hold only in the limit case of $h\rightarrow0$, which is shown in \cite{Solenthaler2011} for inviscid Lagrangian SWE. Furthermore, for constant particle number and mass the continuity equation (\ref{Eq12:Lagconsmass}) is automatically fulfilled by reinterpretation of the density as height of the waves \cite{Solenthaler2011}. In this setting (\ref{Eq13:Lagconsmom}) reduces to regular SPH calculations for incompressible Navier-Stokes equations \cite{Monaghan1985} with pressure and viscosity corresponding forcing terms such as an additional term that arises from variations in the sediment.
\begin{remark}
By relying on (\ref{Eq:DensityWater}) we circumvent the additional carry of information by the individual particles and instead relying on density-based calculation in (\ref{Eq:ParticleObjective}).\\
	In (\ref{Eq:ParticleObjective}) and in what follows we are relying on a constant particle mass, in line with Section \ref{Sec:AdjointParticleSystems}.
	Furthermore, we omit discretization indices for readability whenever it is clear from the context.
\end{remark}

\subsection{Basics of Shape Optimization}
\label{SubSec:BasicsShapeOpt}
In this section we introduce methodologies commonly used in shape optimization, extensively elaborated in monographs \cite{Choi1987}\cite{Sokolowski1992}\cite{Delfour2011}. We hereby mainly follow \cite{Berggren2010}. We start by introducing a family of mappings $\{\phi_\eps\}_{\eps\in[0,\tau]}$ for $\tau>0$ that are used to map each current position $\x\in\Om$ to another by $\phi_\eps(\x)$, where we choose the sufficiently smooth vector field $\Vv$ as the direction for the so-called perturbation of identity
\begin{align}
\x_\eps=\phi_\eps(\x)=\x+\eps \Vv(\x)\text{.}
\label{Eq:4poi}
\end{align}
According to this methodology, we can map the whole domain $\Om$ to another $\Omeps$ such that 
\begin{align}
\Omeps=\{\x_\eps|x+\eps \Vv(x),x\in\Om\}\text{.}
\label{Eq:5domain}
\end{align}
We define the Eulerian Derivative as 
\begin{align}
DJ(\Om)[\Vv]=\lim_{\eps\rightarrow 0^+} \frac{J(\Om_\eps)-J(\Om)}{\eps}\text{.}\label{Eq:7EulerDer}
\end{align}
Commonly, this expression is called shape derivative of $J$ at $\Om$ in direction $\Vv$ and in this sense $J$ shape differentiable at $\Om$ if for all directions $\Vv$ the Eulerian derivative exists and the mapping $\Vv\mapsto DJ(\Om)[\Vv]$ is linear and continuous.
In addition, we define the material derivative of some scalar function $p:\Om\rightarrow\R$ at $x\in\Om$ by the derivative of a composed function $p_\eps\circ\phi_\eps:\Om\rightarrow\Om_\eps\rightarrow\R$ for $p_\eps:\Om_\eps\rightarrow\R$ as
\begin{align}
D_m p(x):=\lim_{\eps\rightarrow0^+}\frac{p_\eps\circ \phi_\eps(x)-p(x)}{\eps}=\frac{d}{d\eps}\restr{(p_\eps\circ \phi_\eps)(x)}{\eps=0^+}\label{Eq:8MatDer}
\end{align} 
and the corresponding shape derivative for a scalar $p$ as
\begin{align}
D p[\Vv]:=D_mp-\Vv\cdot\nab p \label{Eq:9MatDer2}
.
\end{align}
In Section \ref{SubSec:DerSDSPH} we will need the product rule, i.e. \cite{Berggren2010}
\begin{align}
D_m(pq)&=D_mpq+pD_mq\label{Eq:10MatProdR}
\end{align}
and the fact that material derivatives do not commute with spatial derivatives \cite{Berggren2010}
\begin{align}
D_m\nab p&=\nab D_mp-\nab \Vv^T\nab p\label{Eq:11MatGradR}
.
\end{align}

As mentioned in Section \ref{Sec:MoFo} we are relying on a discretize-then-differentiate approach. Hence,  the perturbation of identity w.r.t. a single vertex perturbation $\delta x_l$ for vertices $l\in\{1,...,L\}$ is defined as \cite{Berggren2010}
	\begin{equation}
	\begin{aligned}
	x_\eps&=\phi_\eps(x)\\
	&=x+\eps\delta x_lN_l^1(x)\\
	&=x+\eps\Vv_l(x)\text{,}
	\label{Eq:DiscretizedPerturbationofIdentity}
	\end{aligned}
	\end{equation}
where $N_l^1$ is a finite-element basis function of first order, i.e. a continuous piecewise-linear polynomial. 
\begin{remark}
	The ultimate position $x_\eps\in\Om_\eps$ can be calculated using summed vertex contributions, i.e. 
	\begin{align}
	x_\eps=x+\eps\sum_{l=1}^L\delta x_lN_l^1(x)\text{.}
	\end{align}
	The discretized perturbation of identity interpolates the deformation $\eps\delta x_l$ of vertex $l\in\{1,...,L\}$ on the support of $N^1_l$. 
\end{remark}

\subsection{Adjoint-Based Shape Optimization for SPH Particles}
\label{SubSec:DerSDSPH}
The sensitivity with respect to domain deformations is obtained from Theorem \ref{Theo:AdjointsParticle} as
\begin{equation}
	\begin{aligned}
	DJ_1(\Om)[\Vv_l]&=\restr{\frac{dJ_{1,h}(\Om_\eps)}{d\eps}}{\eps=0^+}\\
	&=\frac{d}{d\eps}\restr{\sum_{k=1}^{n}J_{1,h,k}(x_k,v_k,\Om_\eps)}{\eps=0^+}+\frac{d}{d\eps}\restr{\sum_{k=1}^n(F_{k-1}^{\eps})^T\left(\frac{\Delta t}{m}\mu_k+\frac{\Delta t^2}{m}\delta_k\right)}{\eps=0^+}\text{,}
	\end{aligned}
\end{equation}
where occurring functions are defined on the perturbed domain $\Om_\eps\subset \R^d$.
In need of deriving the shape derivative, we first state the following lemma, following \cite{Berggren2010}.
\begin{lemma}\label{Lemma:ParticleSD}
	For a finite element function $g$, which is an element of a conformal approximation space of $H^1(\Om)$, i.e. $g(x)=\sum_{s=1}^Sg_sN_s^p(x)$ for finite element ansatz function $N_s^p$, whose restriction on an Lagrangian element $\kappa$ is a polynomial of order $p\geq1$, the shape derivative is derived as 
	\begin{align}
	D g[\Vv_l]=\sum_{s=1}^{S}(D g_s[\Vv_l]N_s^p-g_s\Vv_l\cdot\nab N_s^p)\text{,}
	\end{align}
	where $N_s^p(x_\eps)=N_s^p(\tau_\eps^{-1}(x_\eps))$ is moving along the deformation.
\end{lemma}
\begin{proof}\cite{Berggren2010}
	Since $N_s^p(x)$ is moving along the deformation, the material derivative vanishes 
	\begin{align*}
	D_m N_s^p=0\text{,}
	\end{align*}
	hence for the shape derivative it holds by (\ref{Eq:9MatDer2})
	\begin{align*}
	D N_s^p[\Vv_l]=-\Vv_l\cdot\nab N_s^p\text{.}
	\end{align*} 
	Since $g_s$ is spatially constant, we conclude according to \cite{Berggren2010}
	\begin{align*}
	D_m g=\sum_{s=1}^{S}D g_s[\Vv_l]N_s^p
	\label{Eq:ParticleMaterialDerivative}
	\end{align*}
	and therefore obtain the shape derivative as
	\begin{align*}
	D g[\Vv_l]=D_m g-\Vv_l\nab g=\sum_{s=1}^{S}\left(D g_s[\Vv_l]N_s^p-g_s\Vv_l\cdot \nab N_s^p\right)
	\end{align*}
\end{proof}
\begin{remark}
	For SPH boundary computations we recall the shape derivative of the signed distance function \cite{Allaire2016} for $x\notin\Sigma$
	\begin{align}
	D d_\Om(x)[\Vv]=-\Vv(p_{\partial\Om}(x))\cdot\nv(p_{\partial\Om}(x)).
	\end{align}
	with operator $p_{\partial\Om}$ that projects a point $x\in\Om$ onto its closest boundary, where $\Sigma$ is referred to as the ridge. In this sense, $D d_{\Om,s}[\Vv_l]$ is its nodal discretized counterpart w.r.t. the $l^{th}$ vertex perturbation.
\end{remark}

Since the boundary contribution is driven by surface gradient forces, we require the shape derivative of the gradient of a finite element function (\ref{Eq:ParticleBoundaryPressure}).
\begin{lemma}\label{Lemma:ParticleSDGrad}
	For the gradient of a finite element function, i.e. $\nab g(x)=\sum_{s=1}^{S}g_s\nab N_s^p(x)$ for finite element ansatz function $N_s^p$, whose restriction on an Lagrangian element $\kappa$ is a polynomial of order $p\geq2$, the shape derivative is derived as 
	\begin{equation}
	\begin{aligned}
	D (\nab g)[\Vv_l]=\sum_{s=1}^{S}D g_s[\Vv_l]\nab N_s^p-\left(\nab\Vv_l\right)^T\left(\sum_{s=1}^{S}g_s\nab N_s^p\right)-\nab\left(\sum_{s=1}^{S}g_s\nab N_s^p\right)\Vv_l\text{.}
	\end{aligned}
	\end{equation}
\end{lemma}
\begin{proof}
	The material derivative does not commute with the spatial derivative (\ref{Eq:11MatGradR})
	\begin{align}
	D_m(\nab g)=\nab(D_m g)-\left(\nab \Vv_l\right)^T\nab g,
	\end{align}
	which equals by the same argument as in proof to Lemma \ref{Lemma:ParticleSD}
	\begin{align}
	D_m(\nab g)=\nab\left(\sum_{s=1}^{S}D g_s[\Vv_l]N_s^p\right)-\left(\nab \Vv_l\right)^T\left(\sum_{s=1}^{S}g_s\nab N_s^p\right)\text{.}
	\end{align}
	Then we have for the shape derivative
	\begin{equation}
	\begin{aligned}
	D (\nab g)[\Vv_l]&=D_m(\nab g)-\nab(\nab g)\Vv_l\\
	&=\nab\left(\sum_{s=1}^{S}D g_s[\Vv_l]N_s^p\right)- \left(\nab\Vv_l\right)^T \left(\sum_{s=1}^{S}g_s\nab N_s^p\right)-\nab\left(\sum_{s=1}^{S}g_s\nab N_s^p\right)\Vv_l\text{.}
	\end{aligned}
	\end{equation}
\end{proof}
\begin{theorem}(Shape Derivative)\label{Theo:ParticleSD}
	Assume $\{x_k\}_{k=1}^{n}$ moves alongside the deformation, the shape derivative of objective $J_1(\Om)$ is then given by 
	\begin{equation}
	\begin{aligned}
		DJ_1(\Om)[\Vv_l]=&\frac{d}{d\eps}\restr{\sum_{k=1}^{n}J_{1,k}(\x_k,\uv_k,\Om_t)}{\eps=0^+}+\frac{d}{d\eps}\restr{\sum_{k=1}^n(F_{k-1}^{\eps})^T\left(\frac{\Delta t}{m}\mu_k+\frac{\Delta t^2}{m}\delta_k\right)}{\eps=0^+}\\
	=&\sum_{k=1}^n\sum_{i=1}^N\Bigg[\nab\left(\sum_{s=1}^{S}D \gamma_{s}[\Vv_l]N_s^p(\x^i_k)\right)-\left(\nab \Vv_l(\x^i_k)\right)^T \left(\sum_{s=1}^{S}\gamma_{s}\nab N_s^p(\x^i_k)\right)\\
	&\hspace{1cm}-\nab\left(\sum_{s=1}^{S}\gamma_{s}\nab N_s^p(\x^i_k)\right)\Vv_l(\x^i_k)\Bigg]^T\left(\Delta t\mu_k^i+\Delta t^2\delta_k^i\right)\text{.}
	\label{Eq:SDSDF}
	\end{aligned}
	\end{equation}
	where $D\gamma_{s}[\Vv_l]$ is the nodal discretization w.r.t. the $l^{th}$ vertex perturbation of
	\begin{align}
	D \gamma(x)[\Vv]=\text{max}'_\alpha(d_\Om(x))(-\Vv(p_{\partial\Om}(x))\cdot\nv(p_{\partial\Om}(x)))\text{.}
	\end{align}
\end{theorem}
\begin{proof}
	Since  $\Ga$ is fixed, the objective (\ref{Eq:ParticleObjective}) is independent of mesh deformations. Hence, the shape derivative is zero, i.e.
	\begin{align*}
	\frac{d}{d\eps}\restr{\sum_{k=1}^{n}J_{1,k}(\x_k,\uv_k,\Om_\eps)}{\eps=0^+}=0\text{.}
	\end{align*}
	We recall that for SPH-flows the fluid and boundary part are split up (\ref{Eq:FluidBoundarySplit}), i.e. for density of particle $i\in\{1,...,N\}$ in each time step $k\in\{1,...,n\}$
	\begin{align*}
	\rho(x^i_k)=\sum_{j=1}^{N}mW(x^{i}_k-x_k^{j},h)+\sum_{s=1}^{S}mN_{s}(x_k^i)\gamma_{s}+\sum_{s=1}^{S}mN_{s}(x_k^i)\gamma_{s}^L\text{.}
	\end{align*}
	The first sum is only dependent on the position of the particles, hence the respective shape derivatives vanish due to adjoints and we have
	\begin{align*}
	\frac{d}{d\eps}\restr{\sum_{j=1}^{N}mW(x_k^{i}-x_k^{j},h)}{\eps=0^+}=0\text{.}
	\end{align*}
	The remainders follow from regarding terms w.r.t. the water gradient (\ref{Eq:ParticleBoundaryPressure}), Lemma \ref{Lemma:ParticleSD} and \ref{Lemma:ParticleSDGrad} such as equation (\ref{Eq:SDSDF}), by assuming that the particle position $\{x_k^i\}_{i,k=1}^{N,n}$ moves alongside the deformation in each time step.
\end{proof}
\begin{remark}
	The assumption that $\{x_k^i\}_{i,k=1}^{N,n}$ moves along the deformation drastically simplifies calculations, otherwise we are required to perform low-level computations similar to works in \cite{Schneider2008}. 
\end{remark}
\begin{remark}
	Using (\ref{Eq:DiscretizedPerturbationofIdentity}) and the fact that nodal values are spatially constant we can rewrite (\ref{Eq:SDSDF}) as 
	\begin{equation}
	\begin{aligned}
	DJ_1(\Om)[\Vv_l]=&\frac{d}{d\eps}\sum_{k=1}^n\restr{(F_{k-1}^{\eps})^T\left(\frac{\Delta t}{m}\mu_k+\frac{\Delta t^2}{m}\delta_k\right)}{\eps=0^+}\\
	=&\sum_{k=1}^n\sum_{i=1}^N\Bigg[\left(\sum_{s=1}^{S}D \gamma_{s}[\Vv_l]\nab N_s^p(x^i_k)\right)- \left(\delta x_l\nab N_l^1(x^i_k)\right)^T \left(\sum_{s=1}^{S}\gamma_{s}\nab N_s^p(x^i_k)\right)\\
	&\hspace{1cm}-\left(\sum_{s=1}^{S}\gamma_{s}\nab\nab N_s^p(x^i_k)\right)\delta x_lN_l^1(x^i_k)\Bigg]^T\left(\Delta t\mu^i_k+\Delta t^2\delta^i_k\right)\text{.}
	\label{Eq:SDSDF_mod}
	\end{aligned}
	\end{equation}
	Here we highlight, that the product of ansatz functions is zero, whenever, $\{x_k^i\}_{i,k=1}^{N,n}$ is not in the support of shape functions $N_s^p$ and $N_l^1$. 
\end{remark}
	\begin{remark}
	Factoring out the respective $\delta x_l$ in (\ref{Eq:SDSDF_mod}) we can calculate the remaining quantities for each mesh vertex perturbation. All quantities can be collected in a vector $DJ_1(\Om)[\Vv_l]$ of size $dL$ to apply some mesh deformation strategy as discussed in \cite[Section 6.3]{Mohammadi2001}. Linear elasticity calculations as in \cite{Baker1999} can be enabled, by choosing a basis of the test function space, calculating all occurring quantities and split vectorial contributions, as it is common in vector-valued finite element methods.
\end{remark}
	\begin{remark}
	    For completeness shape derivatives of the penalty terms (\ref{Eq:1Vol}) and (\ref{Eq:1Peri}) are obtained as \cite{Sokolowski1992}
        \begin{align}
        DJ_{2}(\Om)[\Vv]&=\nu_1\int_{D}\nab\cdot\Vv\diff x \label{Eq:30DJvol}\\
        DJ_{3}(\Om)[\Vv]&=\nu_2\int_{\Gc}\kappa\langle \Vv,\nv\rangle \diff s=\nu_2\int_{\Gc}\nab\cdot \Vv-\langle \frac{\partial \Vv}{\partial \nv},\nv\rangle \diff s\label{Eq:30DJreg}
        .
        \end{align}
        The discretization of (\ref{Eq:30DJvol}) and (\ref{Eq:30DJreg}) and the subsequent usage in optimization routines follows naturally for standard finite element solvers \cite{Schulz2016}.
	\end{remark}

\subsection{Numerical Results}
In this section we will first give details about the numerical implementation, before verifying results for a selected choice of test cases.\\
The SPH techniques from Section \ref{Sec:SPH} to solve Lagrangian SWE with boundary contributions can be summarized in a pseudocode as

\begin{algorithm}[H]
	%\algsetup{linenosize= \tiny}
	\SetAlgoLined
	%\scriptsize
	Initialize Particles carrying $(\Qv^i_0,x^i_0)_{i=1}^N$ with Constant Mass $m>0$ and Step-Size $\Delta t>0$\\
		\ForEach{Time Step $k\in\{1,...,n\}$}{
	\ForEach{Particle $i$}{
		Find Neighbour Particles $M\subset\{1,...,N\}$
	}
	\ForEach{Particle $i$}{Compute Fluid Density $\rho_F^i$ [via (\ref{Eq:SPHIdea})]\\
				Compute Boundary Density $\rho_D^i$ [via (\ref{Eq:BoundaryFI})]\\
		Compute Water Height $H^i$ [via (\ref{Eq:DensityWater})]}
	\ForEach{Particle $i$}{
		Calculate Forces:\\
		$F^i_k=F_i^{Height}+F_i^{Viscosity}+F_i^{Sediment}$}
	\ForEach{Particle $i$}{Calculate States:\\
		$\Qv^i_{k+1}=\Qv^i_k+\Delta tF^i_k/m$\\
		$\x^i_{k+1}=\x^i_k+\Delta t\Qv^i_{k+1}$}}
	\caption{\small 2D SPH for Lagrangian SWE with Boundary Contribution}
	\label{Algo:SPHFuidParticle}
\end{algorithm}
\begin{remark}
	The following remarks should guide through Algorithm \ref{Algo:SPHFuidParticle}:
	\begin{enumerate}[label=(\roman*)]
		\item The neighbour search is only mentioned for completeness. For the Gaussian kernel (\ref{Eq:GaussianKernel}), the search for neighbouring particles is trivially omitted.
		\item Whenever quantities should be stored for subsequent computations the time step $k\in\{1,...,n\}$ is explicitly mentioned.
		\item The change from Lagrangian SWE to incompressible Navier-Stokes equations is remarkable easy. Instead of computing water heights, a fluid pressure would be required as
		\begin{align}
			p^i_{F}=B_T\left[\left(\frac{\rho^i_{F}}{\rho_0}\right)^\xi-1\right]
		\end{align}
		for $B_T,\xi>0$ e.g. using Tait's law \cite{Becker2007}. 
	\end{enumerate}
\end{remark}
In all the following examples we will work with a simple mesh in line with definitions given in Section \ref{SubSec:DerSDSPH}, as it can be seen in Figure \ref{Fig:ParticleMesh}.
\begin{figure}[htb!]
	\vspace{-1.75cm}
	\centering
	\begin{tikzpicture}
	\node[anchor=south west,inner sep=0](zero){\includegraphics[scale=0.25]{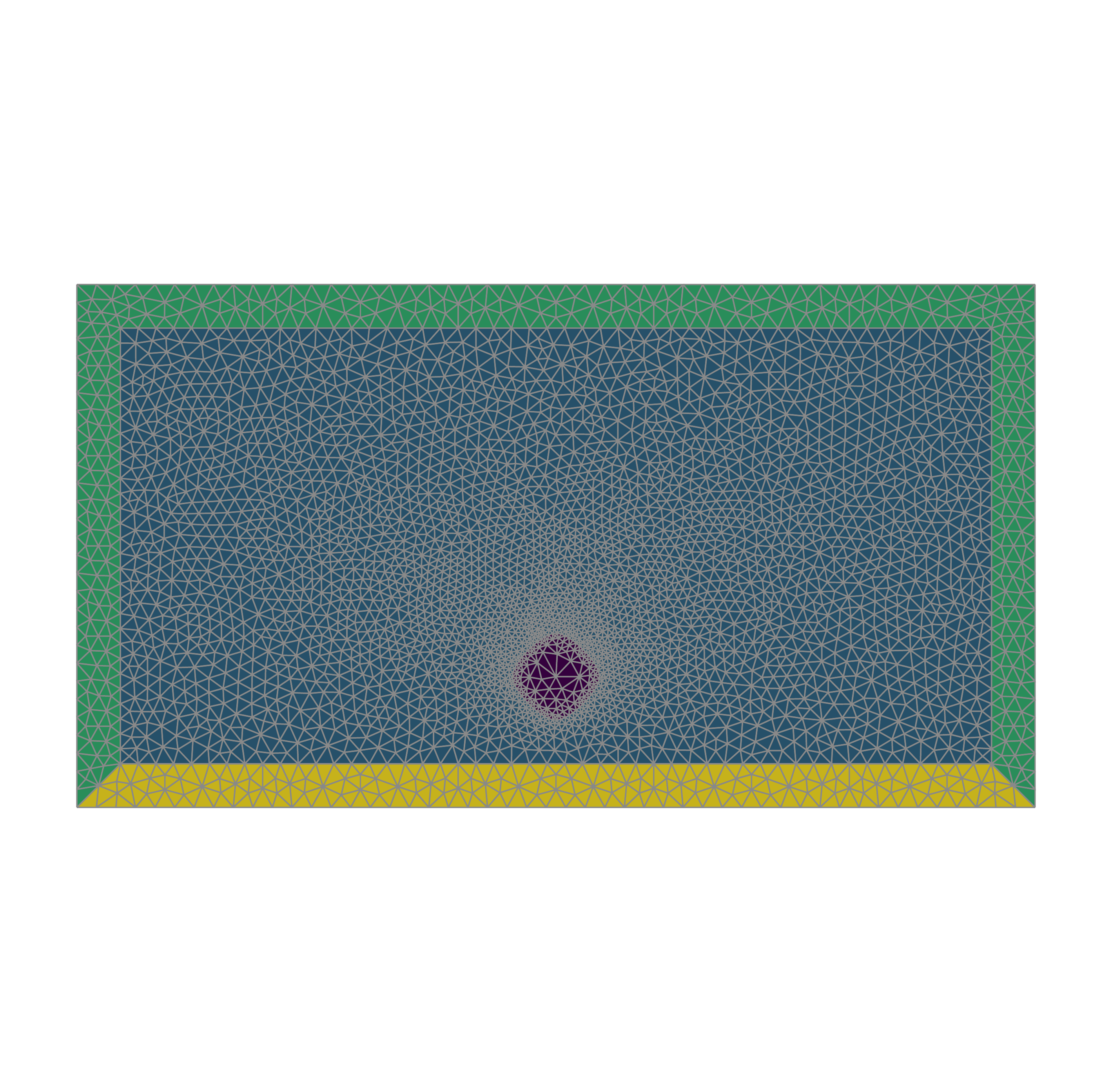}};
	\end{tikzpicture}
	\vspace{-2cm}
	\caption{Initial Mesh with Coloured Subdomains}
	\label{Fig:ParticleMesh}
\end{figure}
The solution to the signed distance function is based on the solution of the Eikonal equation with $f(x)=1$, $q(x)=0$
\begin{equation}
\begin{aligned}
|\nab w(x)|&=f(x) \quad &x\in\Om\\
w(x)&=q(x)\quad &x\in\Gamma
,
\end{aligned}
\end{equation}
where we have implemented a viscous and stabilized version to obtain $w\in H^1(\Om)$ for all $v\in H^1(\Om)$ i.e.
\begin{equation}
\begin{aligned}\int_\Om\sqrt{\nab w\cdot\nab w}v\diff x-\int_\Om fv\diff x +\int_\Om\eps_{SDF} \nab w\cdot\nab v\diff x=0
\label{Eq:ViscousEikonal}
,
\end{aligned}
\end{equation}
where $\eps_{SDF}=\max h_\kappa$ is dependent on the element-diameter $h_\kappa$ of mesh-cell $\kappa\in\mathcal{T}_h$. Build on this solution, modifications as in (\ref{Eq:Smoothed}) are exercised downstream by prescribing respective nodal values.
The recursive adjoints require the calculation of a matrix with partial derivatives w.r.t. the states. The partial derivatives matrix w.r.t. the positions is obtained by a $2N\times 2N$ matrix
\begin{align}
\frac{\partial F}{\partial x}=
\begin{pmatrix}
\frac{\partial F_1^{x}}{\partial x^1} &\frac{\partial F_1^{x}}{\partial x^2} & \ldots & \frac{\partial F_1^{x}}{\partial x^N}&	\frac{\partial F_1^{x}}{\partial y^1} &\frac{\partial F_1^{x}}{\partial y^2} & \ldots & \frac{\partial F_1^{x}}{\partial y^N}\\
\frac{\partial F_2^{x}}{\partial x^1} & \ddots &  & \vdots&\frac{\partial F_2^{x}}{\partial y^1} & \ddots &  & \vdots\\
\vdots  & &\ddots & \vdots&\vdots  & &\ddots & \vdots\\
\frac{\partial F_N^{x}}{\partial x^1} & \ldots & \ldots & \frac{\partial F_N^{x}}{\partial x^N}&\frac{\partial F_N^{x}}{\partial y^1} & \ldots & \ldots & \frac{\partial F_N^{x}}{\partial y^N}\\
\frac{\partial F_1^{y}}{\partial x^1} &\frac{\partial F_1^{y}}{\partial x^2} & \ldots & \frac{\partial F_1^{y}}{\partial x^N}&	\frac{\partial F_1^{y}}{\partial y^1} &\frac{\partial F_1^{y}}{\partial y^2} & \ldots & \frac{\partial F_1^{y}}{\partial y^N}\\
\frac{\partial F_2^{y}}{\partial x^1} & \ddots &  & \vdots&\frac{\partial F_2^{y}}{\partial y^1} & \ddots &  & \vdots\\
\vdots  & &\ddots & \vdots&\vdots  & &\ddots & \vdots\\
\frac{\partial F_N^{y}}{\partial x^1} & \ldots & \ldots & \frac{\partial F_N^{y}}{\partial x^N}&	\frac{\partial F_N^{y}}{\partial y^1} & \ldots & \ldots & \frac{\partial F_N^{y}}{\partial y^N}
\end{pmatrix}
\label{Eq:Pdxmatrix}
\end{align}
and the $2N\times 2N$ matrix of partial derivatives of the particle velocities as
\begin{align}
\frac{\partial F}{\partial u}=
\begin{pmatrix}
\frac{\partial F_1^{x}}{\partial u^1} &\frac{\partial F_1^{x}}{\partial u^2} & \ldots & \frac{\partial F_1^{x}}{\partial u^N}&	\frac{\partial F_1^{x}}{\partial v^1} &\frac{\partial F_1^{x}}{\partial v^2} & \ldots & \frac{\partial F_1^{x}}{\partial v^N}\\
\frac{\partial F_2^{x}}{\partial u^1} & \ddots &  & \vdots&\frac{\partial F_2^{x}}{\partial v^1} & \ddots &  & \vdots\\
\vdots  & &\ddots & \vdots&\vdots  & &\ddots & \vdots\\
\frac{\partial F_N^{x}}{\partial u^1} & \ldots & \ldots & \frac{\partial F_N^{x}}{\partial u^N}&\frac{\partial F_N^{x}}{\partial v^1} & \ldots & \ldots & \frac{\partial F_N^{x}}{\partial v^N}\\
\frac{\partial F_1^{y}}{\partial u^1} &\frac{\partial F_1^{y}}{\partial u^2} & \ldots & \frac{\partial F_1^{y}}{\partial u^N}&	\frac{\partial F_1^{y}}{\partial v^1} &\frac{\partial F_1^{y}}{\partial v^2} & \ldots & \frac{\partial F_1^{y}}{\partial v^N}\\
\frac{\partial F_2^{y}}{\partial u^1} & \ddots &  & \vdots&\frac{\partial F_2^{y}}{\partial v^1} & \ddots &  & \vdots\\
\vdots  & &\ddots & \vdots&\vdots  & &\ddots & \vdots\\
\frac{\partial F_N^{y}}{\partial u^1} & \ldots & \ldots & \frac{\partial F_N^{y}}{\partial u^N}&	\frac{\partial F_N^{y}}{\partial v^1} & \ldots & \ldots & \frac{\partial F_N^{y}}{\partial v^N}\\
\end{pmatrix}\text{.}
\label{Eq:Pdvmatrix}
\end{align}
with two-dimensional components $x=(x,y)$ and $u=(u,v)$.
In the continuous setting, a deformation scheme can be obtained by regarding the Steklov-Poincaré metric \cite{Schulz2016}, such that a solution of the linear elasticity equation $\vec{W}:\Om\rightarrow\R^2$ can be used, i.e.
\begin{equation}
\begin{aligned}
\int\sigma(\vec{W}):\eps(\Vv)&=DJ(\Om)[\Vv] \hspace{1cm} \quad &\forall\Vv\in H_0^1(\Om,\R^2)\\
\sigma:&=\lambda Tr(\eps(\vec{W}))\mathbf{I}_2+2\mu\eps(\vec{W})\\
\eps(\vec{W}):&=\frac{1}{2}(\nab \vec{W}+\nab \vec{W}^T)\\
\eps(\Vv):&=\frac{1}{2}(\nab \Vv+\nab \Vv^T)
,
\label{Eq:29LinearElasticity}
\end{aligned}
\end{equation}
where $\sigma$ and $\eps$ are called strain and stress tensor and $\lambda$ and $\mu$ are called Lamé parameters. We have chosen $\lambda=0$ and $\mu$ as the solution of the following Poisson problem
\begin{equation}
\begin{aligned}
-\bigtriangleup\mu&=0 \hspace{1cm} &\text{in } &\Om&&\\
\mu&=\mu_{max} \hspace{1cm} &\text{on } &\Gc&&\\
\mu&=\mu_{min} \hspace{1cm} &\text{on } &\Ga, \Gb.&&
\end{aligned}
\label{Eq:33Lame}
\end{equation} 
In our discrete setting is the right hand-side in (\ref{Eq:29LinearElasticity}) manually evaluated using (\ref{Eq:SDSDF_mod}) and the associated remark. The full gradient-descent based shape optimization algorithm for SPH-flows is pseudocoded below.
\begin{algorithm}[H]
	\SetAlgoLined
	Initialization Mesh, Particles\\
	\While{$||DJ(\Om_k)[\Vv]||>\eps_{TOL}$}{
		1. Calculate Modified SDF $\gamma_k$ [via Viscous Eikonal Eq. (\ref{Eq:ViscousEikonal})]\\
		2. Calculate States $\x_k,\uv_k$ [via SPH Algorithm \ref{Algo:SPHFuidParticle}]\\
		3. Calculate Adjoints $\delta_k,\mu_k$ [via Section \ref{Sec:AdjointParticleSystems} \& AD]\\
		4. Calculate Gradient $W_k$ [via (\ref{Eq:SDSDF_mod}-\ref{Eq:30DJreg}) \& Linear Elasticity (\ref{Eq:29LinearElasticity})]\\
		5. Perform Linesearch for $\tilde{W}_k$\\
		6. Calculate $\Om_{k+1}$ [via $\tilde{W}_k$ and (\ref{Eq:DiscretizedPerturbationofIdentity})] }
	\caption{\small Lagrangian SWE via SPH Shape Optimization Algorithm}
\end{algorithm}
\begin{remark}
	Calculation of partial derivatives in (\ref{Eq:Pdxmatrix}) and (\ref{Eq:Pdvmatrix}) is manually possible, however results in a tremendous calculative effort, prone to errors. For circumvention we use automatic differentiation (AD) via the Autograd\footnote[1]{https://github.com/HIPS/autograd} library in either a forward or backward mode for the fluid portion, which is possible for code generated fully in Numpy\footnote[2]{https://numpy.org/}. Here we highlight for boundary terms as defined in (\ref{Eq:BoundaryFI}) the partial derivative matrix for positional state is diagonal and zero for the velocity state. We also point out that derivatives of boundary contributions are fully implemented in FEniCS \cite{FEniCS2015}. The manual evaluation of (\ref{Eq:SDSDF}) in Step 4 is relying on a constant evaluation of finite element ansatz
	functions, where the FEniCS build-in function \textit{evaluate\_basis\_derivatives\_all} is used.
\end{remark}

\subsubsection*{Example 1:}
\label{Sec:NumericalExamplesSDSPH}
In the first example we model the propagation of two particles, i.e. $N=2$, with initial position $x_0^1=(0.5,0.5)$ and $x_0^2=(0.63,0.5)$,
towards the shore by prescribing initial velocities as $\uv_0=(0,-3)^2$. This test case is deliberately kept simple to analyse the procedure. In Figure \ref{Fig:Ex1ParticleProp} we have visualized the particle movements for time snapshots $t\in\{0,0.6,0.12,0.24\}$ for smoothing radius $h=0.04$, particle mass $m=1$, smoothing factor $\alpha=100$ and reference density $\rho_0=10$. The particle propagation is performed using time steps of size $\Delta t=0.008$ with end time $T=0.18$.  The red particle, with initial position $x_0^1$, travels towards the shore $\Ga$ and effectively accounts for an increased objective in form of (\ref{Eq:ParticleObjective}). In contrast, the blue particle contributes only marginally to the objective, starting from initial position $x_0^2$, being reflected from the obstacle and hence travelling back into the field and being collected in the outflow channel.

\begin{figure}[htb!]
	\centering
	\begin{subfigure}{0.4\textwidth}
		\centering
		\begin{tikzpicture}
		\node[anchor=south west,inner sep=0] (1) 
{\includegraphics[scale=0.3]{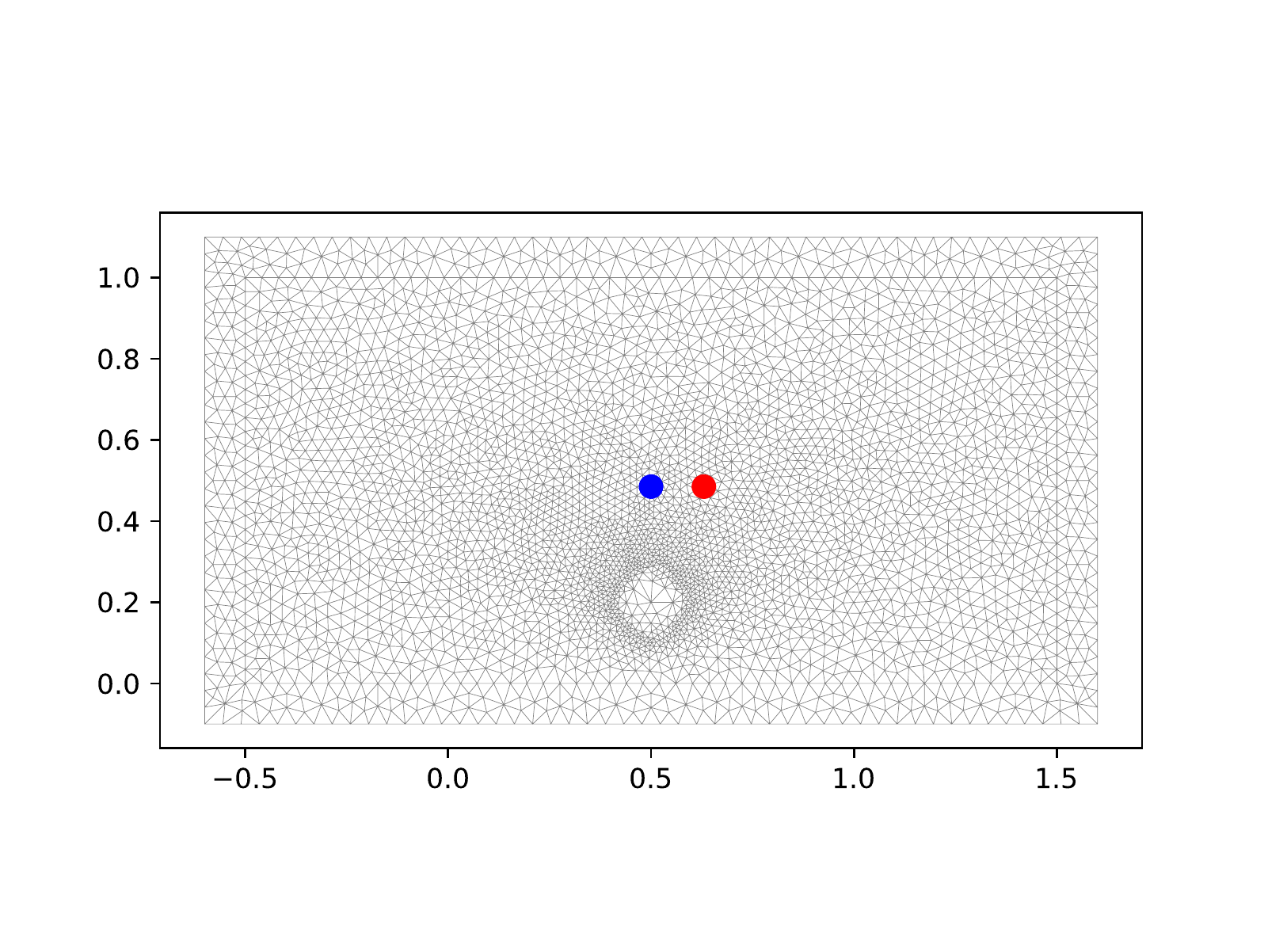}};
		\end{tikzpicture}
		\caption{$t=0$}
	\end{subfigure}
	\hfill
	\begin{subfigure}{0.4\textwidth}
		\centering
		\begin{tikzpicture}
		\node[anchor=south west,inner sep=0] (1) {
	\includegraphics[scale=0.3]{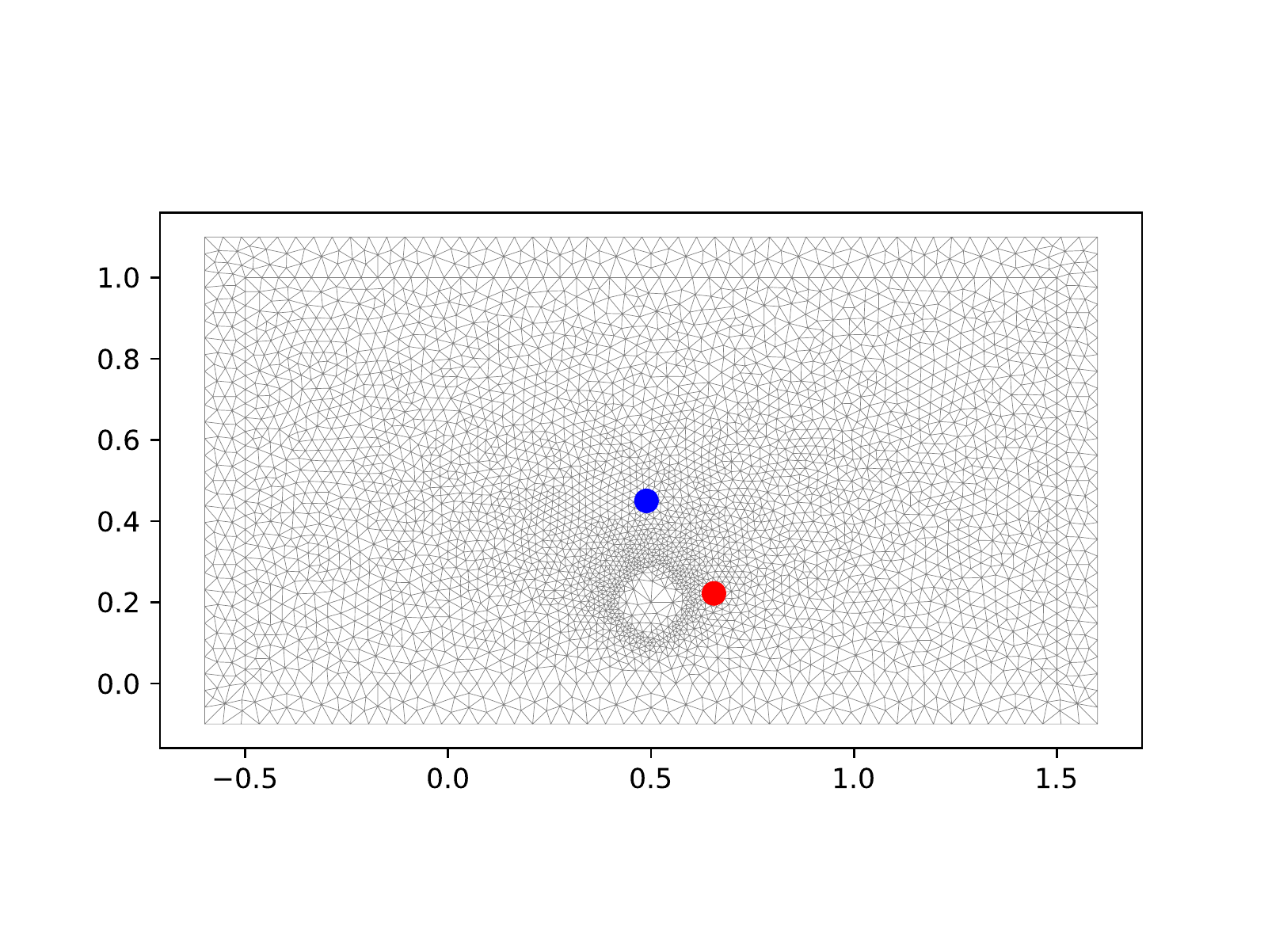}};
		\end{tikzpicture}
		\caption{$t=0.06$}
	\end{subfigure}
	\hfill
	\begin{subfigure}{0.4\textwidth}
		\centering
		\begin{tikzpicture}
		\node[anchor=south west,inner sep=0] (1) 
	{\includegraphics[scale=0.3]{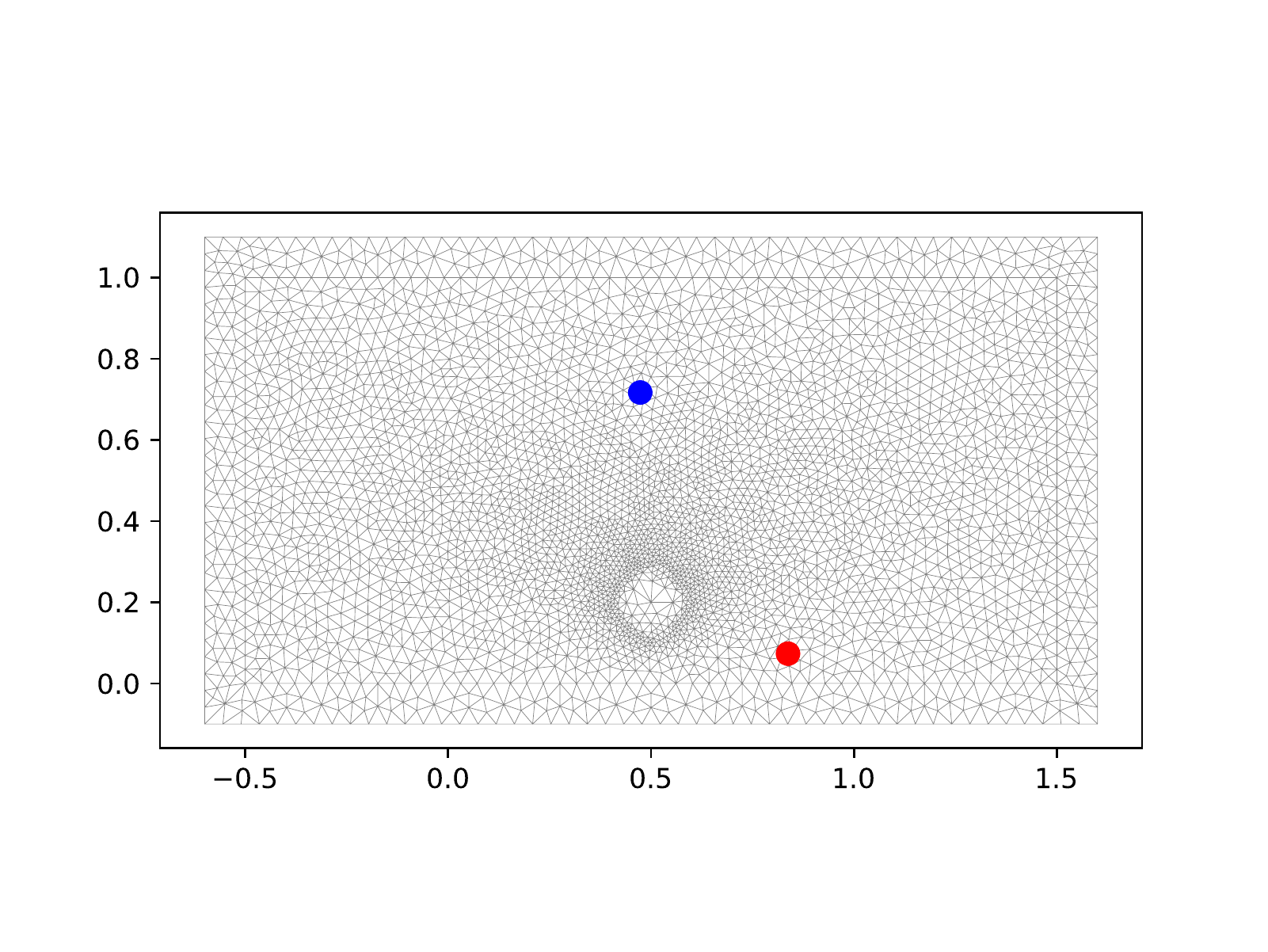}};
		\end{tikzpicture}
		\caption{$t=0.12$}
	\end{subfigure}
	\hfill
	\begin{subfigure}{0.4\textwidth}
		\centering
		\begin{tikzpicture}
		\node[anchor=south west,inner sep=0] (1) 
	{\includegraphics[scale=0.3]{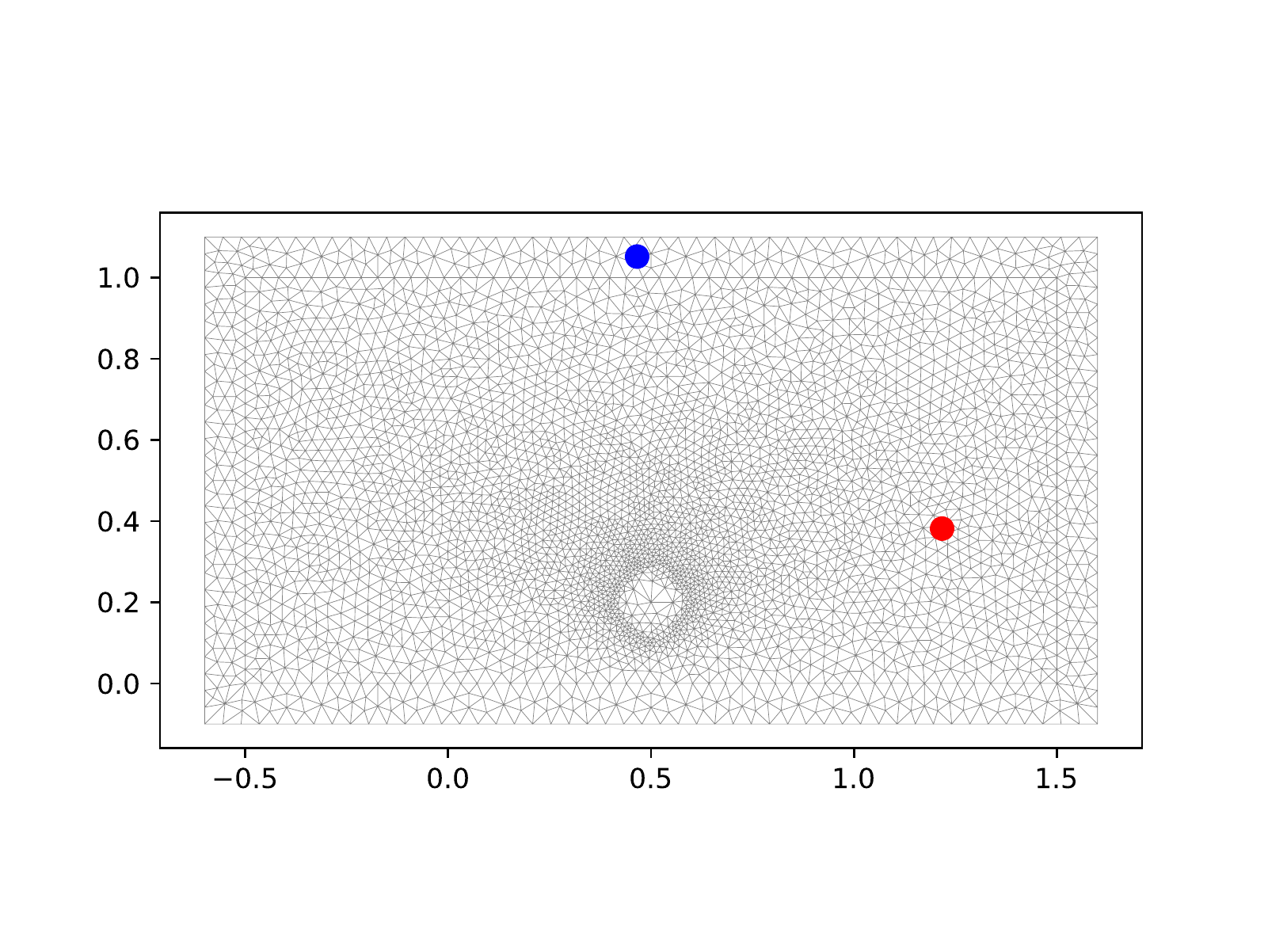}};
		\end{tikzpicture}
		\caption{$t=0.24$}
	\end{subfigure}
	\caption[Ex.1 Particle Propagation]{Ex.1 Particle Propagation}
	\label{Fig:Ex1ParticleProp}
\end{figure}

Both particles are within the interaction radius of the obstacle for a certain time frame, such that shape derivatives of boundary height forces are non-zero. However, as stated above, the objective contribution of blue ensures that adjoints $\{\delta^1_k\}_{k=1}^n$ and $\{\mu^1_k\}_{k=1}^n$ are vanishingly small. The initial deformation vector is hence predominantly activated in reds interaction region with the obstacle, as it can be observed in Figure \ref{Fig:Ex1FinalResults} on the left. Relying on the shape optimization algorithm that we have presented before, for initial step-size $\rho_{step}=\expnumber{1}{-4}$ and tolerance $\eps_{TOL}=\expnumber{1}{-9}$, we are able to deform the obstacle setting $\mu_{min}=10$ and $\mu_{max}=100$ in the Poisson problem. The final mesh is able to decrease the objective up to a minimum, as it can be seen in Figure \ref{Fig:Ex1FinalResults} in the mid and right part. We highlight, that even a small deformation is able to achieve these results in this simple setting.
\begin{figure}[htb!]
	\centering
	\begin{subfigure}{0.3\textwidth}
		\centering
		\begin{tikzpicture}
		\node[anchor=south west,inner sep=0] (1) 
		{\includegraphics[scale=0.3]{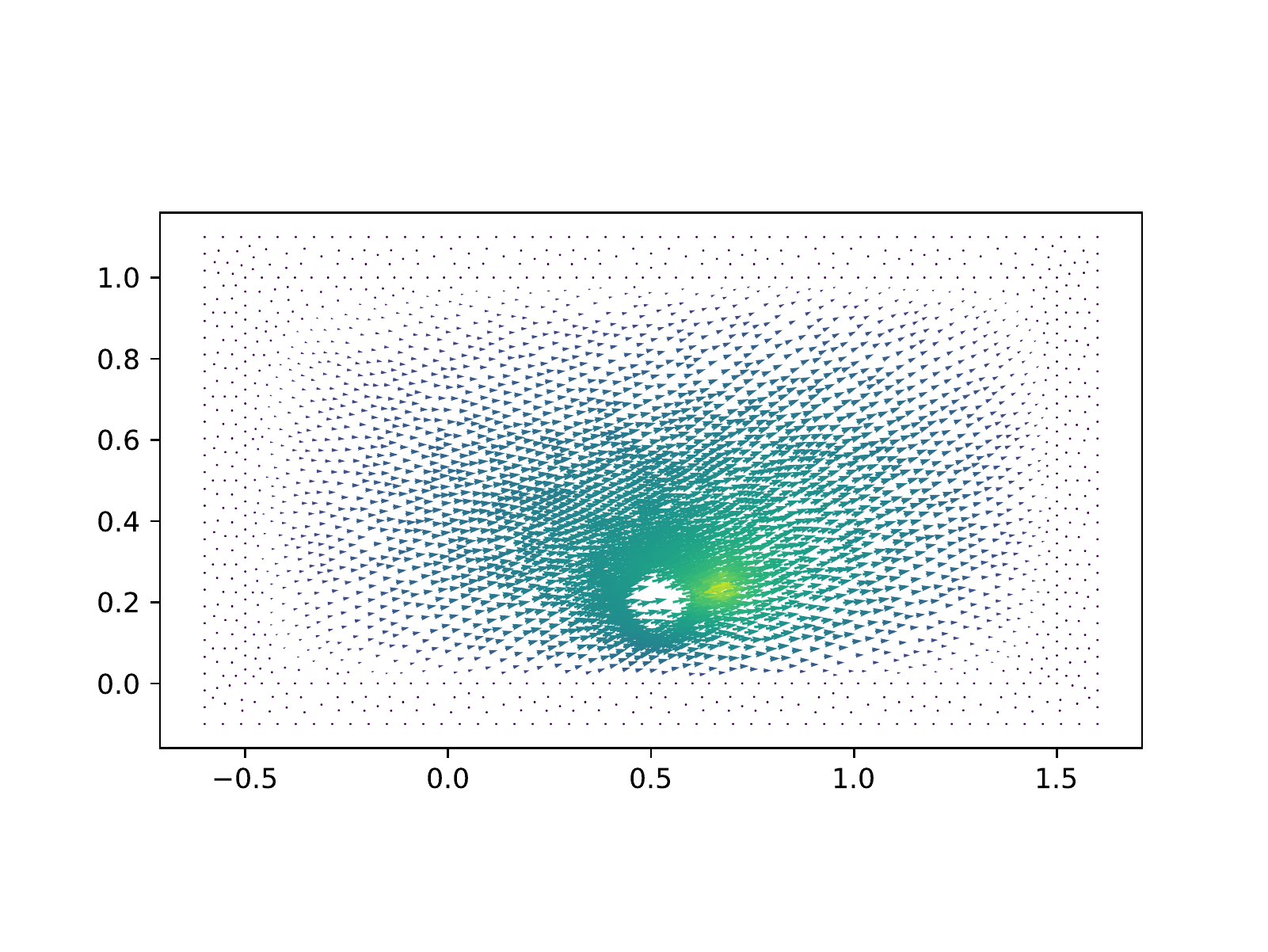}};
		\end{tikzpicture}
		\caption{Deformation Field}
	\end{subfigure}
	\hfill
	\begin{subfigure}{0.3\textwidth}
		\centering
		\begin{tikzpicture}
		\node[anchor=south west,inner sep=0] (1) {
			\includegraphics[scale=0.3]{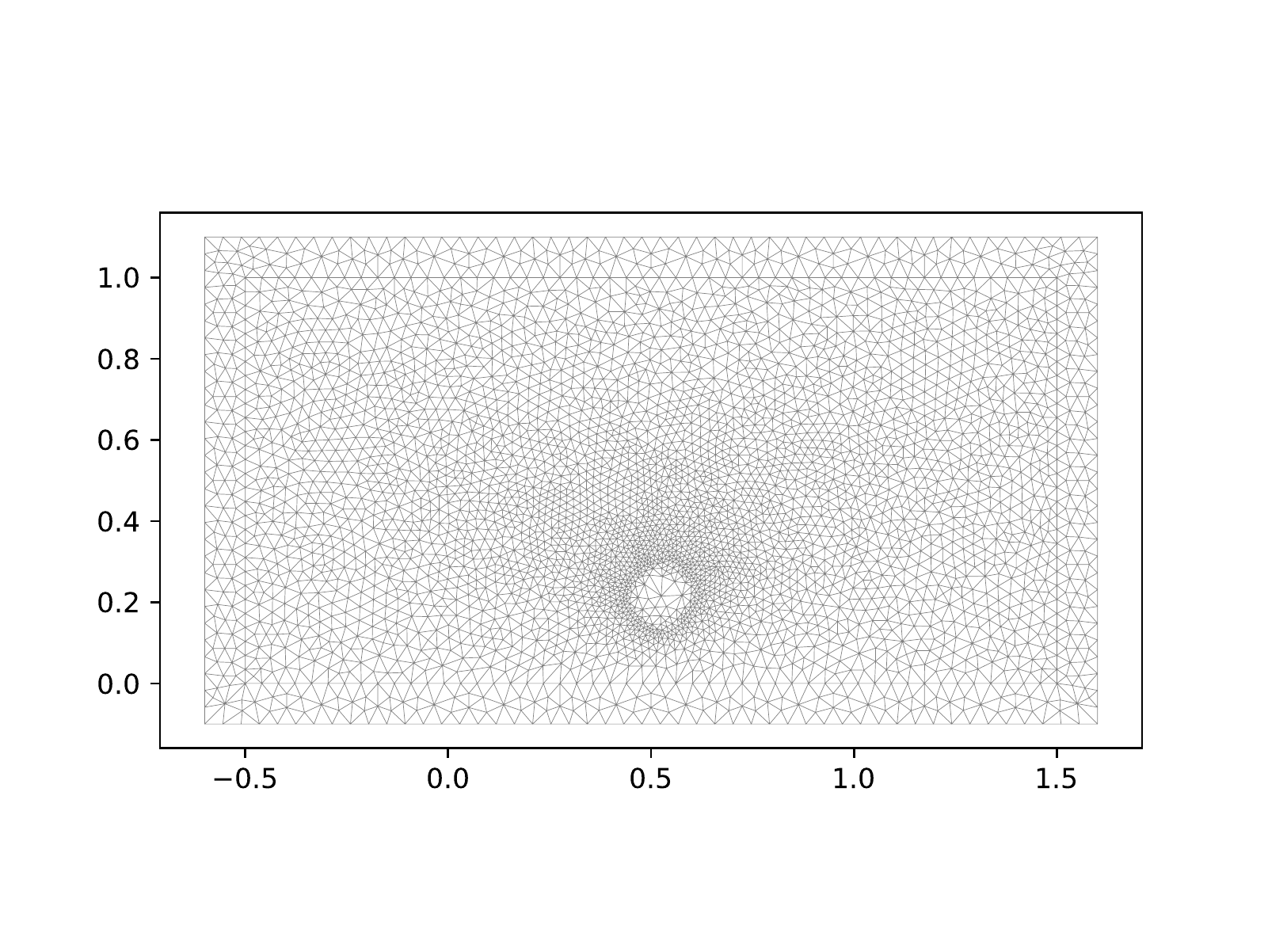}};
		\end{tikzpicture}
		\caption{Optimized Mesh}
	\end{subfigure}
	\hfill
	\begin{subfigure}{0.3\textwidth}
		\centering
		\begin{tikzpicture}
		\node[anchor=south west,inner sep=0] (1) 
	{\scalebox{0.48}{% This file was created with tikzplotlib v0.9.13.
\begin{tikzpicture}

\begin{axis}[
legend cell align={left},
legend style={fill opacity=0.8, draw opacity=1, text opacity=1, draw=white!80!black},
tick align=outside,
tick pos=left,
title={Objective Value},
x grid style={white!69.0196078431373!black},
xlabel={Iteration},
xmajorgrids,
xmin=-0.25, xmax=19.25,
xtick style={color=black},
y grid style={white!69.0196078431373!black},
ylabel={Objective},
ymajorgrids,
ymin=-0.0752199117504898, ymax=1.57961814705944,
ytick style={color=black}
]
\addplot [line width=1.64pt, blue]
table {%
0 1.50439823529535
1 1.49342343568817
2 1.48254666452127
3 1.46256127844682
4 1.43476869766745
5 0.941223121211458
6 0.419317835417319
7 0.378111248810039
8 0.348345430178122
9 0.319423482115772
10 0.291481255879910
11 0.2601648484878
12 0.2284522698912
13 0.151215778118328
14 3.48775122318848e-02
15 1.15823211123287e-04
16 1.76160047800943e-06
17 1.35978252070067e-11
18  1.2329323793408e-12
19  1.8121421079965e-13
};
\addlegendentry{$J(\Omega)$}
\end{axis}

\end{tikzpicture}}};
		\end{tikzpicture}
		\caption{Objective}
	\end{subfigure}
	\caption[Ex.1 Optimization Results]{Ex.1 Optimization Results}
	\label{Fig:Ex1FinalResults}
\end{figure}
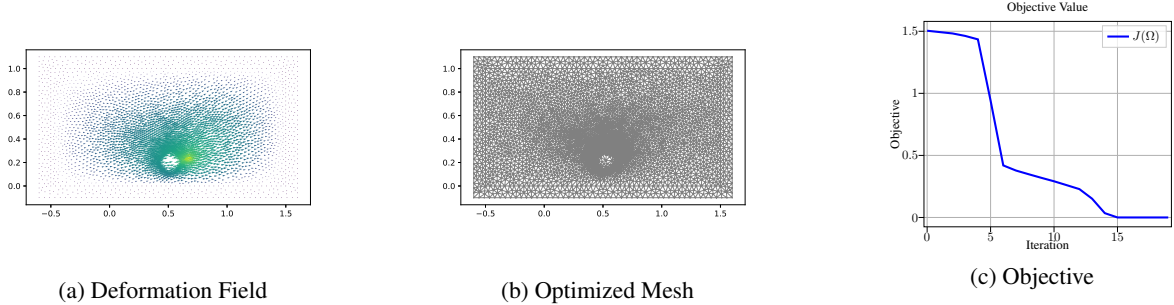
\FloatBarrier

\subsubsection*{Example 2:}
In the second example we increase the number of particles, i.e. $N=180$, with initial positions drawn from a normal distribution as $x_0=(0.7+\mathcal{N}(0,0.1),0.6+\mathcal{N}(0,0.1))^{N}$. The particles once more travel towards the shore driven by initial velocity as $\uv_0=(0,-3)^N$. The remaining settings are as in the first example. We can observe the movement in Figure \ref{Fig:Ex2ParticleProp}.

\begin{figure}[htb!]
	\centering
	\begin{subfigure}{0.4\textwidth}
		\centering
		\begin{tikzpicture}
		\node[anchor=south west,inner sep=0] (1) 
		{\includegraphics[scale=0.3]{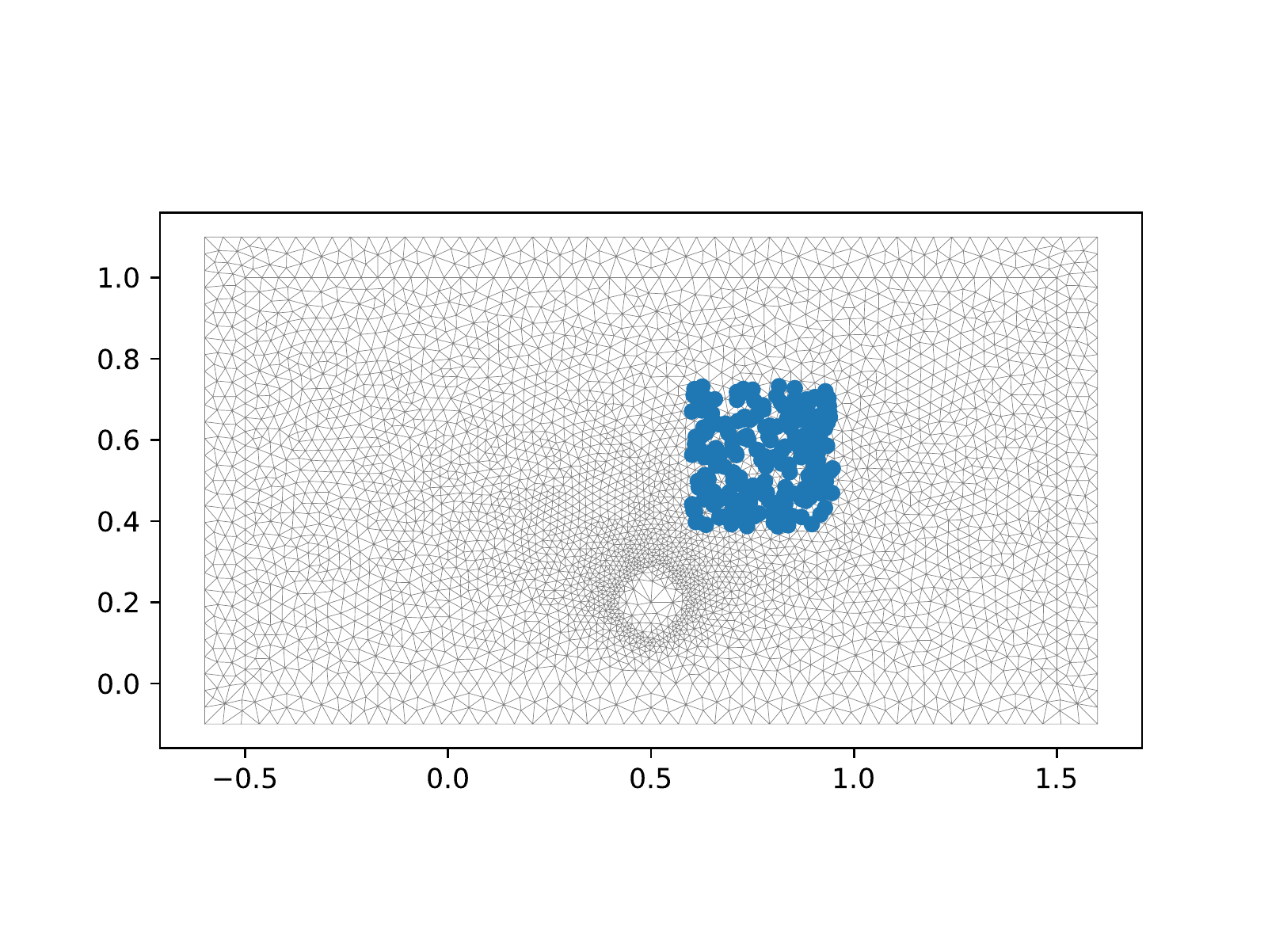}};
		\end{tikzpicture}
		\caption{$t=0$}
	\end{subfigure}
	\hfill
	\begin{subfigure}{0.4\textwidth}
		\centering
		\begin{tikzpicture}
		\node[anchor=south west,inner sep=0] (1) {
			\includegraphics[scale=0.3]{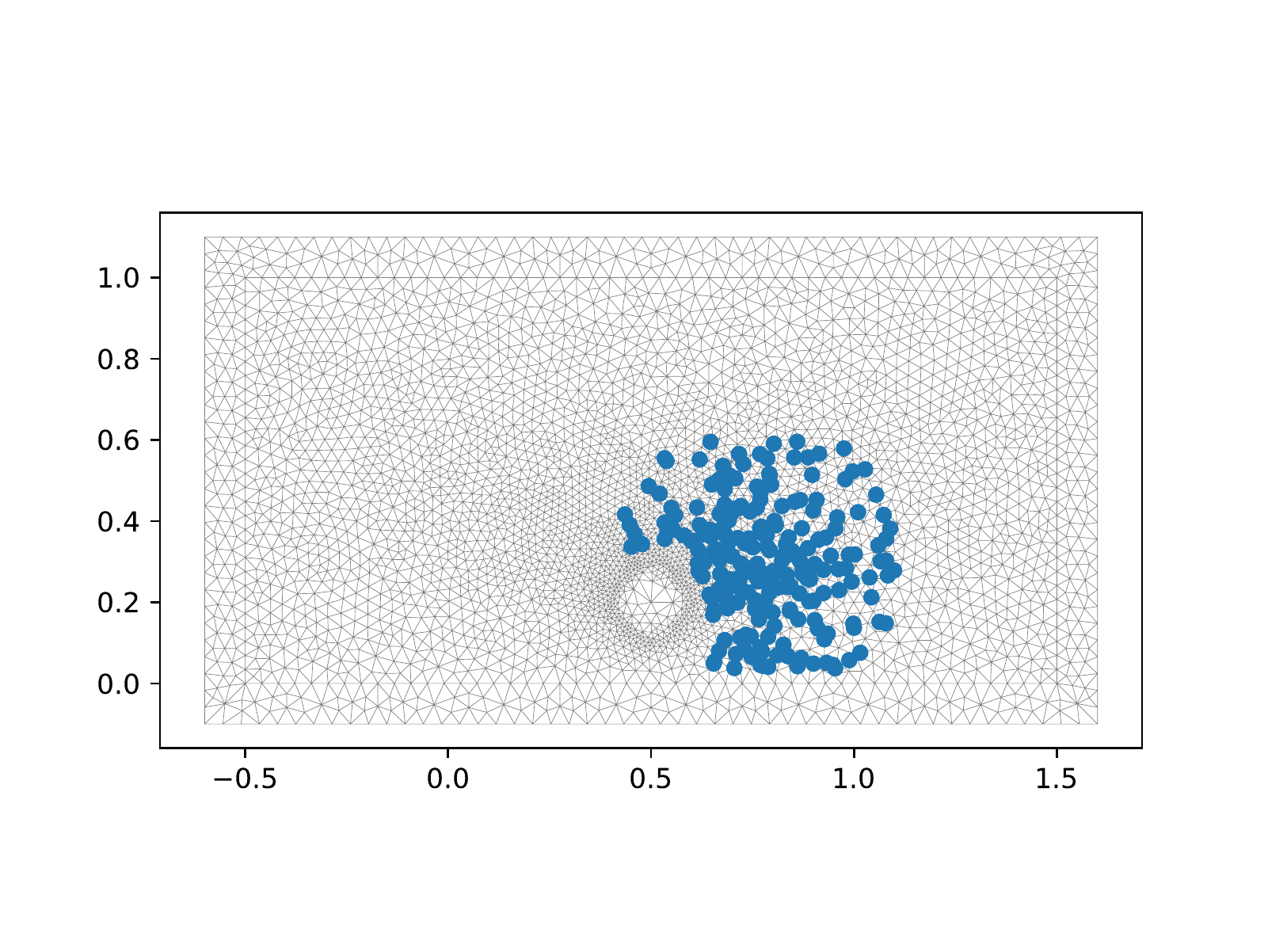}};
		\end{tikzpicture}
		\caption{$t=0.06$}
	\end{subfigure}
	\hfill
	\begin{subfigure}{0.4\textwidth}
		\centering
		\begin{tikzpicture}
		\node[anchor=south west,inner sep=0] (1) 
		{\includegraphics[scale=0.3]{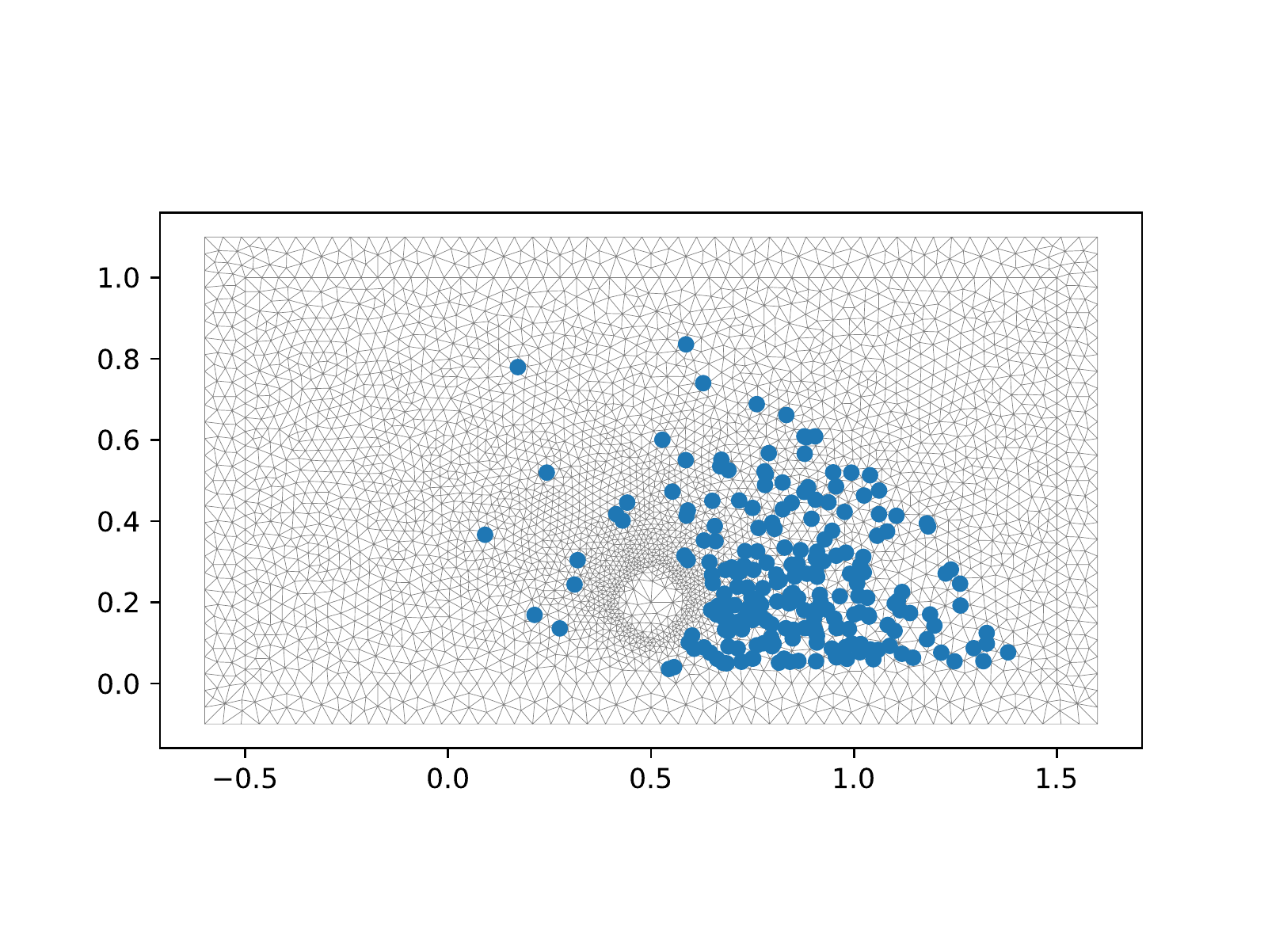}};
		\end{tikzpicture}
		\caption{$t=0.12$}
	\end{subfigure}
	\hfill
	\begin{subfigure}{0.4\textwidth}
		\centering
		\begin{tikzpicture}
		\node[anchor=south west,inner sep=0] (1) 
		{\includegraphics[scale=0.3]{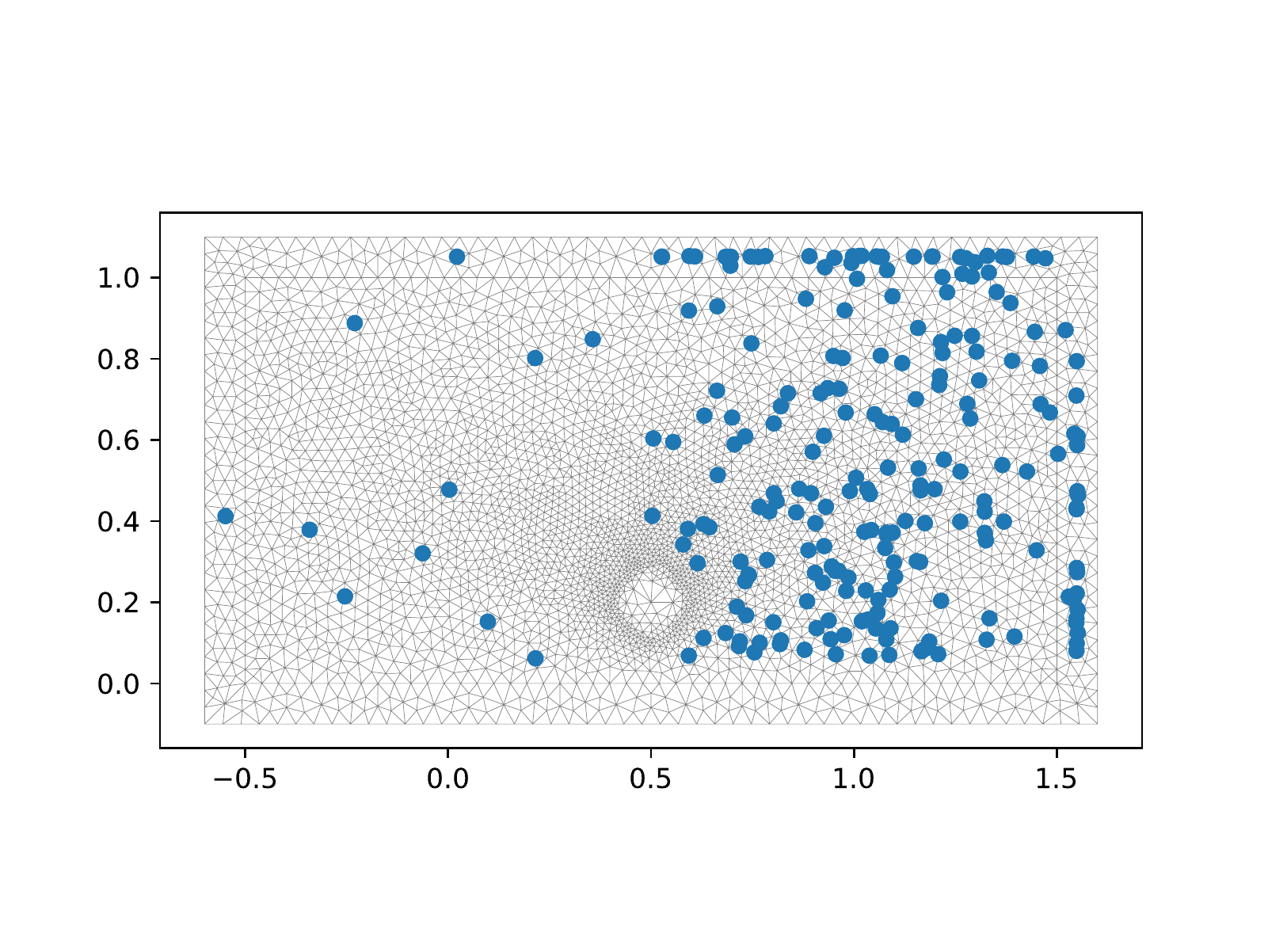}};
		\end{tikzpicture}
		\caption{$t=0.24$}
	\end{subfigure}
	\caption[Ex.2 Particle Propagation]{Ex.2 Particle Propagation}
	\label{Fig:Ex2ParticleProp}
\end{figure}
Once more relying on our shape optimization algorithm, we are able to deform the obstacle using resulting deformation fields, pictured in Figure \ref{Fig:Ex2FinalResults} on the left for the first iteration. In this setting a larger deformation is necessary to effectively reduce the objective as we see in the final mesh Figure \ref{Fig:Ex2FinalResults} in the middle and on the right. We would like to highlight that the performance of the algorithm can be degrading by consecutive particle and boundary interactions, resulting from mesh deformations. Potentially this hinders us from obtaining improved results.
\begin{figure}[htb!]
	\centering
		\begin{subfigure}{0.3\textwidth}
		\centering
		\begin{tikzpicture}
		\node[anchor=south west,inner sep=0] (1) 
		{\includegraphics[scale=0.3]{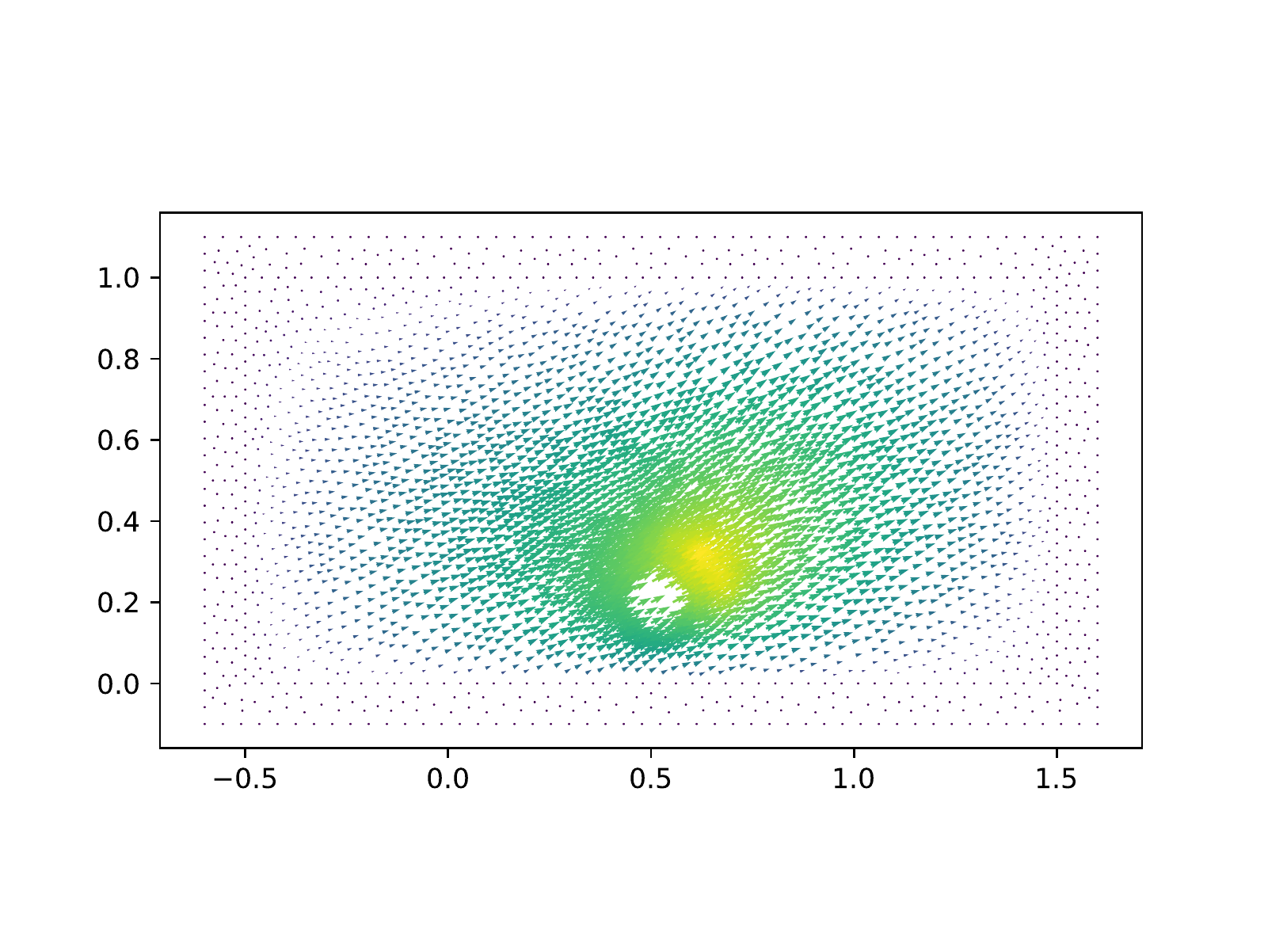}};
		\end{tikzpicture}
		\caption{Deformation Field}
	\end{subfigure}
	\hfill
	\begin{subfigure}{0.3\textwidth}
		\centering
		\begin{tikzpicture}
		\node[anchor=south west,inner sep=0] (1) {
			\includegraphics[scale=0.3]{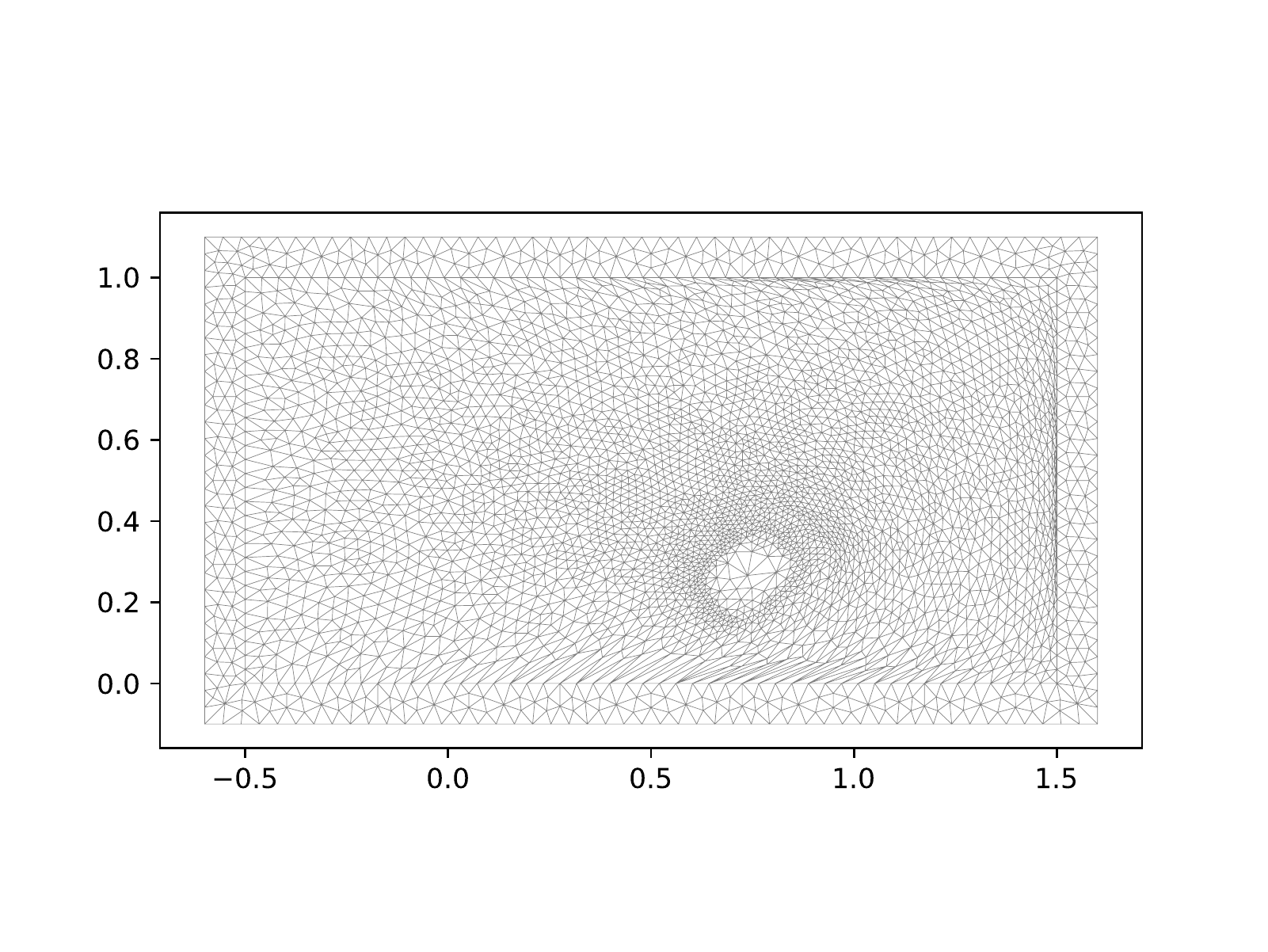}};
		\end{tikzpicture}
		\caption{Optimized Mesh}
	\end{subfigure}
	\hfill
	\begin{subfigure}{0.3\textwidth}
		\centering
		\begin{tikzpicture}
		\node[anchor=south west,inner sep=0] (1) 
		{\scalebox{0.48}{% This file was created with tikzplotlib v0.9.13.
\begin{tikzpicture}

\begin{axis}[
legend cell align={left},
legend style={fill opacity=0.8, draw opacity=1, text opacity=1, draw=white!80!black},
tick align=outside,
tick pos=left,
title={Objective Value},
x grid style={white!69.0196078431373!black},
xlabel={Iteration},
xmajorgrids,
xmin=-1.3, xmax=39.3,
xtick style={color=black},
y grid style={white!69.0196078431373!black},
ylabel={Objective},
ymajorgrids,
ymin=0.086929073001349, ymax=0.765186609257832,
ytick style={color=black}
]
\addplot [line width=1.64pt, blue]
table {%
0 7.343645123863469237e-01
1 7.274266265165393319e-01
2 7.255383333648666166e-01
3 7.188610711322995783e-01
4 7.130755929948184857e-01
5 7.104677807974565384e-01
6 7.073330877446863951e-01
7 6.977039395345019379e-01
8 6.776743659001573938e-01
9 6.476947673952653250e-01
10 6.122958238457260061e-01
11 5.779328995567558636e-01
12 5.501791468046994815e-01
13 5.321145519309712357e-01
14 5.218957316027424964e-01
15 5.139419383558284649e-01
16 5.006333107082631129e-01
17 4.825847623296318223e-01
18 4.648836216072309568e-01
19 4.462749287180571289e-01
20 4.236087576193520166e-01
21 3.894853486477977134e-01
22 3.462172893251926342e-01
23 2.954838987924565674e-01
24 2.441286081935567231e-01
25 2.007528922071740163e-01
26 1.681968959077970727e-01
27 1.475704052961760437e-01
28 1.368762707836694847e-01
29 1.320635410573941515e-01
30 1.288071761746039656e-01
31 1.243767258727496161e-01
32 1.188039619468131969e-01
34 1.168712121555231765e-01
35 1.145998797209151765e-01
36 1.145689736869864861e-01
37 1.145577204875489748e-01
38 1.145552432538854509e-01
};
\addlegendentry{$J(\Omega)$}
\end{axis}

\end{tikzpicture}}};
		\end{tikzpicture}
		\caption{Objective}
	\end{subfigure}
	\caption[Ex.2 Optimization Results]{Ex.2 Optimization Results}
	\label{Fig:Ex2FinalResults}
\end{figure}
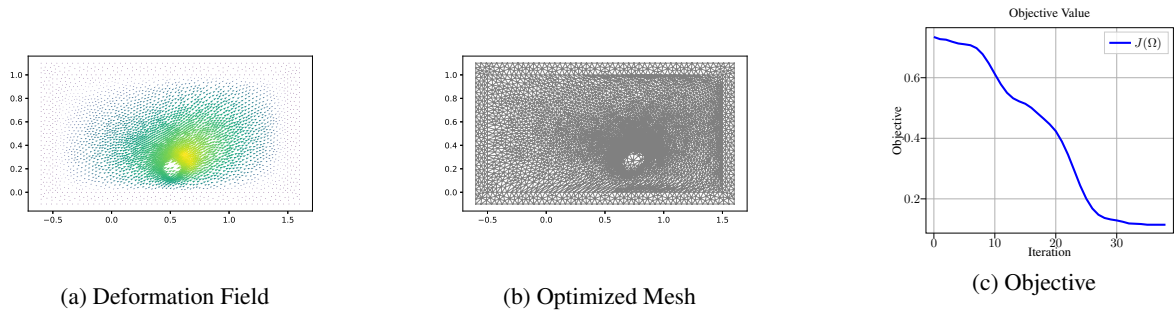

\section{Conclusion}
We have derived the discrete adjoint for a general symplectic Euler particle class. In addition, for an SPH-representantive, based on boundary interactions by signed distance fields, the shape derivative was derived, that reduces to the shape derivative of a finite element interpolator. Results have been verified on a sample mesh inspired by practical applications to mitigate effects of coastal erosion. We point out that this can only serve as a first feasibility study, since the number of particles is still low and meshes are simplified. In future, comparisons between shape optimization results in an Eulerian, cf. to \cite{Schlegel20212} and \cite{Schlegel_2022}, and an Lagrangian framework, such as an Eulerian-Lagrangian coupling e.g. for a SWE-Exner model, appear of interest.

\section*{Acknowledgement}
This work has been supported by the Deutsche
Forschungsgemeinschaft within the Priority program SPP 1962 "Non-smooth and Complementarity-based Distributed Parameter Systems: Simulation and Hierarchical Optimization". The authors would like to thank Diaraf Seck (Université Cheikh Anta Diop, Dakar, Senegal) and Mame Gor Ngom (Université Cheikh Anta Diop, Dakar, Senegal) for helpful and interesting discussions within the project Shape Optimization Mitigating Coastal Erosion (SOMICE). 

\bibliographystyle{unsrt}
\bibliography{bibliography}  %%% Uncomment this line and comment out the ``thebibliography'' section below to use the external .bib file (using bibtex) .

\begin{thebibliography}{10}

\bibitem{Nayroles1992}
Bernard Nayroles, Gilbert Touzot, and Pierre Villon.
\newblock Generalizing the finite element method: Diffuse approximation and
  diffuse elements.
\newblock {\em Computational Mechanics}, 10:307--318, 1992.

\bibitem{Belytschko1994}
T.~Belytschko, Y.~Y. Lu, and L.~Gu.
\newblock Element-free galerkin methods.
\newblock {\em International Journal for Numerical Methods in Engineering},
  37(2):229--256, 1994.

\bibitem{Liu1995}
Wing~Kam Liu, Sukky Jun, and Yi~Fei Zhang.
\newblock Reproducing kernel particle methods.
\newblock {\em International Journal for Numerical Methods in Fluids},
  20(8-9):1081--1106, 1995.

\bibitem{Koshizuka1996}
Seiichi Koshizuka and Yoshiaki Oka.
\newblock Moving-particle semi-implicit method for fragmentation of
  incompressible fluid.
\newblock {\em Nuclear Science and Engineering}, 123:421--434, 1996.

\bibitem{Ogami1991}
Yoshifumi Ogami and Teruaki Akamatsu.
\newblock Viscous flow simulation using the discrete vortex model—the
  diffusion velocity method.
\newblock {\em Computers \& Fluids}, 19(3):433--441, 1991.

\bibitem{Monaghan2002}
Joseph Monaghan.
\newblock Sph compressible turbulence.
\newblock {\em Monthly Notices of the Royal Astronomical Society}, 335, 04
  2002.

\bibitem{Bobaru2002}
Florin Bobaru.
\newblock Meshless approach to shape optimization of linear thermoelastic
  solids.
\newblock {\em International Journal for Numerical Methods in Engineering},
  53:765 -- 796, 02 2002.

\bibitem{Grindeanu1998}
Iulian Grindeanu, Kuang-Hua Chang, Jiun-Shyan Chen, and Kyung Choi.
\newblock Design sensitivity analysis of hyperelastic structures using a
  meshless method.
\newblock {\em Aiaa Journal - AIAA J}, 36:618--627, 04 1998.

\bibitem{Chen2007}
J.~Chen and Nam Kim.
\newblock {\em Meshfree Method and Application to Shape Optimization}, pages
  389--414.
\newblock World Scientific, 09 2007.

\bibitem{Hohmann2019}
Raphael Hohmann and Christian Leithäuser.
\newblock Gradient-based shape optimization for the reduction of particle
  erosion in bended pipes, 2019.

\bibitem{Ha2011}
Youn~Doh Ha, Min-GEUN Kim, Hyun-seok Kim, and Seonho Cho.
\newblock Shape design optimization of sph fluid–structure interactions
  considering geometrically exact interfaces.
\newblock {\em Structural and Multidisciplinary Optimization}, 44:319--336, 09
  2011.

\bibitem{Koschier2017}
Dan Koschier and Jan Bender.
\newblock Density maps for improved sph boundary handling.
\newblock {\em Proceedings of the ACM SIGGRAPH / Eurographics Symposium on
  Computer Animation}, 2017.

\bibitem{Berggren2010}
Martin Berggren.
\newblock A unified discrete-continuous sensitivity analysis method for shape
  optimization.
\newblock In {\em CSC 2010}, 2010.

\bibitem{Schneider2008}
René Schneider and Peter Jimack.
\newblock On the evaluation of finite element sensitivities to nodal
  coordinates.
\newblock {\em Electronic Transactions on Numerical Analysis. Volume},
  32:134--144, 01 2008.

\bibitem{Schulz2016}
Volker~H. Schulz, Martin. Siebenborn, and Kathrin. Welker.
\newblock Efficient pde constrained shape optimization based on
  steklov--poincaré-type metrics.
\newblock {\em SIAM Journal on Optimization}, 26(4):2800--2819, 2016.

\bibitem{Schulz1996}
Volker Schulz.
\newblock Numerical optimization of the cross-sectional shape of turbine
  blades, 1996.

\bibitem{McNamara2004}
Antoine McNamara, Adrien Treuille, Zoran Popovi\'{c}, and Jos Stam.
\newblock Fluid control using the adjoint method.
\newblock {\em ACM Trans. Graph.}, 23(3):449–456, August 2004.

\bibitem{Wojtan2006}
C.~Wojtan, P.~Mucha, and Greg Turk.
\newblock Keyframe control of complex particle systems using the adjoint
  method.
\newblock In {\em SCA '06}, 2006.

\bibitem{Monaghan1985}
Joseph~John Monaghan and John Lattanzio.
\newblock A refined particle method for astrophysical problems.
\newblock {\em Astronomy and Astrophysics}, 149:135--143, 1985.

\bibitem{Muller2003}
Matthias M{\"u}ller, David Charypar, and Markus~H. Gross.
\newblock Particle-based fluid simulation for interactive applications.
\newblock In {\em SCA '03}, 2003.

\bibitem{Raviart1985}
P.~A. Raviart.
\newblock An analysis of particle methods.
\newblock In Franco Brezzi, editor, {\em Numerical Methods in Fluid Dynamics},
  pages 243--324, Berlin, Heidelberg, 1985. Springer Berlin Heidelberg.

\bibitem{Moussa2000}
B.~Ben~Moussa and J.~P. Vila.
\newblock Convergence of sph method for scalar nonlinear conservation laws.
\newblock {\em SIAM Journal on Numerical Analysis}, 37(3):863--887, 2000.

\bibitem{Dilisio1998}
R.~{Di Lisio}, E.~Grenier, and M.~Pulvirenti.
\newblock The convergence of the sph method.
\newblock {\em Computers \& Mathematics with Applications}, 35(1):95--102,
  1998.

\bibitem{Oelschlaeger1991}
K.~Oelschl{\"a}ger.
\newblock {\em On the Connection Between Hamiltonian Many-particle Systems and
  the Hydrodynamical Equations}.
\newblock Universit{\"a}t Heidelberg. SFB 123, 1991.

\bibitem{Kulsegaram2004}
Sivakumar Kulasegaram, Javier Bonet, Roland Lewis, and Matthew Profit.
\newblock A variational formulation based contact algorithm for rigid
  boundaries in 2d sph applications.
\newblock {\em Computational Mechanics}, 33:316--325, 03 2004.

\bibitem{Vacondio2012}
Renato Vacondio, Benedict Rogers, P.K. Stansby, and Paolo Mignosa.
\newblock Sph modeling of shallow flow with open boundaries for practical flood
  simulation.
\newblock {\em Journal of Hydraulic Engineering}, 138:530--541, 06 2012.

\bibitem{Christof2017}
Constantin Christof, Christian Clason, Christian Meyer, and Stephan Walther.
\newblock Optimal control of a non-smooth semilinear elliptic equation.
\newblock {\em Mathematical Control and Related Fields (MCRF)}, 8:247--276, 05
  2017.

\bibitem{Reyes2011}
Juan De~los Reyes.
\newblock Optimal control of a class of variational inequalities of the second
  kind.
\newblock {\em SIAM J. Control and Optimization}, 49:1629--1658, 07 2011.

\bibitem{Solenthaler2011}
Barbara Solenthaler, Peter Bucher, Nuttapong Chentanez, Matthias Müller, and
  Markus Gross.
\newblock {SPH Based Shallow Water Simulation}.
\newblock In Jan Bender, Kenny Erleben, and Eric Galin, editors, {\em Workshop
  in Virtual Reality Interactions and Physical Simulation "VRIPHYS" (2011)}.
  The Eurographics Association, 2011.

\bibitem{Choi1987}
Kyung~K. Choi.
\newblock Shape design sensitivity analysis and optimal design of structural
  systems.
\newblock In Carlos~A. Mota~Soares, editor, {\em Computer Aided Optimal Design:
  Structural and Mechanical Systems}, pages 439--492. Springer Berlin
  Heidelberg, 1987.

\bibitem{Sokolowski1992}
J.~Soko{\l}owski and J.P. Zol{\'e}sio.
\newblock {\em Introduction to Shape Optimization: Shape Sensitivity Analysis}.
\newblock Springer series in computational mathematics. Springer-Verlag, 1992.

\bibitem{Delfour2011}
M.~C. Delfour and J.~P. Zolésio.
\newblock {\em Shapes and Geometries}.
\newblock Society for Industrial and Applied Mathematics, second edition, 2011.

\bibitem{Allaire2016}
Grégoire Allaire, François Jouve, and Georgios Michailidis.
\newblock Thickness control in structural optimization via a level set method.
\newblock {\em Structural and Multidisciplinary Optimization}, 53, 06 2016.

\bibitem{Mohammadi2001}
Bijan Mohammadi and Olivier Pironneau.
\newblock Applied shape optimization in fluids.
\newblock {\em Applied Shape Optimization for Fluids}, 05 2001.

\bibitem{Baker1999}
Timothy~J. Baker and Peter~A. Cavallo.
\newblock Dynamic adaptation for deforming tetrahedral meshes.
\newblock In {\em 14th Computational Fluid Dynamics Conference}, 1999.

\bibitem{Becker2007}
Markus Becker and Matthias Teschner.
\newblock Weakly compressible sph for free surface flows.
\newblock In Dimitris Metaxas and Jovan Popovic, editors, {\em
  Eurographics/SIGGRAPH Symposium on Computer Animation}. The Eurographics
  Association, 2007.

\bibitem{FEniCS2015}
Martin~S. Aln{\ae}s, Jan Blechta, Johan Hake, August Johansson, Benjamin
  Kehlet, Anders Logg, Chris Richardson, Johannes Ring, Marie~E. Rognes, and
  Garth~N. Wells.
\newblock The fenics project version 1.5.
\newblock {\em Archive of Numerical Software}, 3(100), 2015.

\bibitem{Schlegel20212}
Luka Schlegel and Volker Schulz.
\newblock Shape optimization for the mitigation of coastal erosion via shallow
  water equations, 2021.

\bibitem{Schlegel_2022}
Luka Schlegel and Volker Schulz.
\newblock Shape optimization for the mitigation of coastal erosion via porous
  shallow water equations.
\newblock {\em International Journal for Numerical Methods in Engineering},
  2022.

\end{thebibliography}





% \end{thebibliography}

\end{document}